\documentclass[11pt,reqno]{article}
\usepackage[round]{natbib}

\usepackage{amsmath,amsfonts,amsthm}
\usepackage[margin=1in]{geometry}
\usepackage[hyphens]{url}
\usepackage{hyperref}
\usepackage{graphicx}
\usepackage{verbatim}
\usepackage{enumitem} 
\usepackage{color}

\numberwithin{equation}{section}
\theoremstyle{plain}
\newtheorem{Definition}{Definition}[section]
\newtheorem{Remark}{Remark}[section]
\newtheorem{Theorem}{Theorem}[section]
\newtheorem{Lemma}{Lemma}[section]
\newtheorem{Proposition}{Proposition}[section]

\newtheorem{Assumption}{Assumption}[section]

\theoremstyle{definition} 
\newtheorem{Example}{Example}[section]
\newtheorem*{continuedex}{Example \continuedexref\space (Continued)}

\newenvironment{examcont}[1]
  {\newcommand{\continuedexref}{\ref*{#1}}\continuedex}
  {\endcontinuedex}

\def \E{\mathbb{E}}
\def \F{\mathbb{F}}
\def \N{\mathbb{N}}
\def \P{\mathbb{P}}
\def \Q{\mathbb{Q}}
\def \R{\mathbb{R}}
\def \X{\mathbb{X}}

\def \Bc{{\mathcal B}}

\def \Ec{{\mathcal E}}
\def \Fc{{\mathcal F}}

\def \Lc{{\mathcal L}}

\def \Tc{{\mathcal T}}
\def \Qc{{\mathcal Q}}
\def \eps{\varepsilon}

\def \1{\mathds{1}}

\def \Pro#1{\mathbb{P}\left[{#1}\right]}

\begin{document}
\title{General Stopping Behaviors of Na\"{i}ve and Non-Committed Sophisticated Agents,
	with Application to Probability Distortion\thanks{
Y.-J. Huang gratefully acknowledges financial support from National Science Foundation (DMS-1715439)
and a start-up grant at the University of Colorado (11003573).
A. Nguyen-Huu has received support from the Energy \& Prosperity Chair, Risk Foundation, and Labex Entreprendre.
X. Y. Zhou gratefully acknowledges financial support through a start-up grant at Columbia University and through Oxford--Nie Financial Big Data Lab.}
\thanks{The authors would like to thank the two anonymous referees, whose careful reading and thoughtful suggestions substantially enhanced the quality of this paper.
}}
\author{
Yu-Jui Huang\thanks{
Department of Applied Mathematics, University of Colorado, Boulder, CO 80309, USA, \texttt{yujui.huang@colorado.edu}}
\and
Adrien Nguyen-Huu\thanks{
CEE-M, Univ Montpellier, CNRS , INRA, Montpellier SupAgro, Montpellier, France, \texttt{adrien.nguyen-huu@umontpellier.fr}
 }
\and
Xun Yu Zhou\thanks{
Department of Industrial Engineering and Operations Research, Columbia University, New York, New York 10027, USA, \texttt{xz2574@columbia.edu}}
}

\maketitle

\begin{abstract}
We consider the problem of stopping a diffusion process with a payoff functional 
that renders the problem time-inconsistent.
We study stopping decisions  of na\"{i}ve agents who reoptimize continuously in time, as well as equilibrium strategies of sophisticated agents who anticipate but lack control over their future selves' behaviors. When the state process is one dimensional and the payoff functional satisfies some regularity conditions, we prove that any equilibrium can be obtained as a fixed point of an operator. This operator represents strategic reasoning that takes the future selves' behaviors into account.
We then apply the general results to the case when the agents distort probability
and the diffusion process is a geometric Brownian motion. The problem is
inherently time-inconsistent as the level of distortion of a same event changes over time. We show how the strategic reasoning may turn a na\"{i}ve agent into a sophisticated one. Moreover, we derive stopping strategies of the two types of agent for various parameter specifications of the problem, illustrating rich behaviors beyond the extreme ones such as ``never-stopping" or ``never-starting".
\end{abstract}

\smallskip

\smallskip


\section{Introduction}
\label{sec:introduction}

Optimal stopping is to determine the best (random) time to stop a stochastic process
so as to maximize a given payoff arising from the stopped state of the process.
Applications of such a timing problem are abundant including stock trading (e.g.,  the best time to sell a stock),
option pricing (e.g.,  American options) and real options (e.g.,  the best time to invest in a project).
Two main classical approaches to solving  optimal stopping problems are dynamic programming and martingale theory,
which are both based foundationally on the assumption of time-consistency, namely,
any optimal stopping rule planned today remains optimal tomorrow.

The assumption of time-consistency (also called dynamic-consistency) is rooted in the premise that
an individual's preferences are consistent over time and will not change as time goes by or circumstances evolve.
However, this premise is all too vulnerable to stand the test of reality.
A gambler may have planned initially to leave the casino after having made \$1000,
but then decides to go on gambling when this target has been reached.
This is because the gambler's risk preference has changed after winning \$1000;
he has become more risk-seeking as the result of the triumph.
Indeed, extensive empirical and experimental studies all point to the fact that
time-inconsistency is the rule rather than the exception.\footnote{Refer to the classical papers \cite{S55}, \cite{KP77} for detailed discussions on time-inconsistency, as well as \cite{Thaler81}, \cite{LT89} for empirical and experimental evidence.}

In the absence of time-consistency,  whatever an ``optimal" plan derived at this moment is generally not optimal evaluated at the next moment;
hence there is no such notion of  a ``dynamically optimal strategy" good for the entire planning horizon as is the case with a time-consistent model.
Economists, starting from \cite{S55},
instead have described the responses or behaviors of three types of individuals when facing a time-inconsistent situation.
Type 1, the {\it na\"{i}ve agent}, chooses whatever seems to be optimal with respect to his current preference,
without knowing the fact that his preferences may change over time.
Type 2 is the {\it pre-committed agent}, who optimizes only once at the initial time and then sticks to the resulting  plan in the future.
Type 3, the {\it sophisticated agent}, is aware of the fact that his ``future selves" will overturn his current plan (due to the lack of commitment)
and selects the best present action taking the future disobedience as a constraint.
The resulting strategy is called a (subgame perfect) {\it equilibrium} from which no incarnation of the agent has incentive to deviate.

In our view, a time-inconsistent model is largely {\it descriptive} as opposed to its time-consistent counterpart that is mainly {\it prescriptive}. The objective of the former is to describe behaviors of the different types of agents, rather than to advise on the best course of action\footnote{The strategies of the three types of agents coincide if the problem is time consistent.}.

The first objective of this paper is to study the stopping strategies of Type 1 (na\"{i}ve) and Type 3 (sophisticated) agents for a general time-inconsistent problem.\footnote{Pre-committed strategies (i.e., those of Type 2's) have been investigated for specific class of problems, e.g., ones involving probability distortion in \cite{XZ13}. See discussions below.}
A na\"{i}ve stopping law is one such that, at any given time and state,
it coincides with the optimal stopping law at {\it that} particular pair of time and state.
Definition of equilibrium strategies in the continuous-time setting, on the other hand, is more subtle.
Starting with \cite{EL06}, and followed among other by \cite{EL10}, \cite{yong2012time},  \cite{HJZ12}, \cite{BKM17}, and \cite{HJZ17},
an equilibrium for a control problem is defined as a control process that satisfies a first-order inequality condition on some spike variation of the control.
\cite{EWZ2017} apply this definition to a stopping problem by turning the latter into a control problem.
However, it remains a problem to rigorously establish the equivalence between  this first-order condition and
the zeroth-order condition in the original definition of a subgame perfect equilibrium.

In this paper we follow the formulation of \cite{HN16} to define an equilibrium stopping law
(although therein, a special stopping problem with a non-exponential discount factor featuring decreasing impatience is considered).
The idea of this formulation is that, for any given stopping law,
the sophisticated agent improves it by a level of strategic reasoning through anticipating his future selves' behaviors.
The agent then performs additional levels of similar reasoning until he cannot further improve it, which is an equilibrium.
Mathematically, an equilibrium is a fixed-point of an operator that represents one level of this strategic thinking. This in particular coincides with the zeroth-order condition in the original definition of a subgame perfect equilibrium.

The general existence of such a fixed-point is a largely open question.
The first contribution of this paper is to prove that,
assuming that the state follows a one-dimensional diffusion process
and the payoff functional satisfies some (very mild) measurable and Fatou type conditions,
any equilibrium strategy can be obtained by applying the aforementioned operator
repeatedly on an initial stopping law and then taking the limit.

The second objective of the paper is to apply the general results
to the problem when probability is distorted.
Experimental evidence supports the phenomenon of probability distortion (or weighting);
in particular that people inflate the small probabilities of occurrence of both extremely favorable and extremely unfavorable events.
Behavioral theories such as rank-dependent utility (RDU; \cite{Q82,S89}) and cumulative prospect theory (CPT; \cite{TK92})
include probability distortion as a key component and  consider it as a part of the risk preference.\footnote{At the outset the main difference between RDU and CPT is that the outcome utility function is concave in the former while $S$-shaped in the latter. The difference is more fundamental in economic interpretations. In the CPT,
there is a reference point in wealth and the input 
of the utility function (called {\it value function} in CPT) is the deviation of wealth from the reference point. The $S$-shaped value function captures the common psychological bias that people are risk-averse for gains (when the aforementioned deviation is positive) and risk-seeking for losses (when the deviation is negative). Both theories feature
probability distortion which is the cause of time-inconsistency. So  some of the discussions in this paper will be on the two theories simultaneously.}
On the other hand, the level of probability distortion associated with a same future event changes over time
because the conditional probability of that event evolves dynamically and the distortion is nonlinear.
Therefore time-inconsistency is inherent in the presence of  (non-trivial) probability distortion.

We study the stopping problem of
 a geometric Brownian motion and a RDU type of payoff functional involving probability distortion,
 whose pre-committed stopping strategies have been studied thoroughly by \cite{XZ13}
 and are conceptually and technically useful for the present paper.\footnote{Due to time-inconsistency,
	the dynamic programming principle or the martingale approach does not work for deriving pre-committed strategies.
	\cite{XZ13} develop a new approach with a combination of the so-called quantile/distribution formulation
	and Skorokhod's embedding theory to solve the problem.}
  Besides characterizing  the stopping behaviors of the na\"{i}ve and the sophisticated,
 we are particularly interested in the connection and potential transformation among the three type of agents:
 the na\"{i}ve agent in effect solves  a pre-committed problem at every state and time who, interestingly,
 may turn himself into a  sophisticate if he carries out several levels (sometimes just one level) of the strategic reasoning.
 In this regard, the fixed-point characterization of equilibrium strategies brings about a significant advantage over the first-order characterization,
 as the former describes the ``reasoning process" with which an initial na\"{i}ve strategy is changed into a sophisticated one.
This change is demonstrated explicitly in this paper through practical examples motivated by utility maximization and real options valuation.
A related advantage of our fixed-point characterization is that one can easily generate a large class of equilibria through fixed-point iterations. This is in contrast to the first-order characterization in the literature, where finding multiple equilibria is usually a challenge.

\cite{ebert2015until} show that under the assumptions leading to ``skewness preference in the small",
a na\"{i}ve agent following CPT will never stop ``until the bitter end".
These assumptions essentially ensure that probability distortion on small probabilities of big gains outweighs loss aversion;
hence at any wealth level the agent always favors a small,
right-skewed risk over the deterministic payoff resulting from the immediate stop.
The antithesis of this result is presented in \cite{ESnever2017}:
assuming that there is an arbitrarily strong probability distortion on small probabilities of big losses
and considering only two-threshold type of strategies\footnote{With a two-threshold type of strategy,
	the agent stops whenever the state process reaches $a$ or $b$ where $a<b$ are the two prescribed thresholds.
	\cite{ESnever2017} consider pure Markovian strategies,
	which are essentially two-threshold strategies.},
a sophisticate
will stop immediately or, equivalently, ``never, ever getting started".
This is because when the process is sufficiently close to the upper threshold the stopped state under the
two-threshold strategy becomes highly left-skewed,
which is unattractive to a CPT agent with a typical inverse-S shaped distortion function.
Knowing that the future selves who are close to the upper threshold will not carry out this strategy,
in equilibrium each  current self will simply not start it.

The authors of these two papers acknowledge that both the ``never stop" and ``never start"
behaviors are ``extreme" representing ``unrealistic predictions" under the CPT.
In this paper, we complement their results by studying the cases in which the assumptions of \cite{ebert2015until} and \cite{ESnever2017} do not hold, and
showing  that the behaviors of the na\"{i}ve 
and sophisticated agents are far richer than the mere extreme ones. Indeed, depending on the parameters of the geometric Brownian motion and the payoff functional, both types of agents may {\it want} to start and end,
relying on a (one) threshold-type stopping strategy. This suggests that the RDU model may offer
more realistic prediction if we allow the distortion functions to take various shapes and, meanwhile, we take
into consideration the intertwining relationship between the preference of the agent and the characteristics of the process he monitors.

The rest of the paper is organized as follows. In Section 2 we formulate a general time-inconsistent stopping problem, define the na\"ive and sophisticated equilibrium stopping laws, and introduce the operator that represents a strategic reasoning of a sophisticated agent.  Section 3 characterizes the set of equilibrium stopping laws when the state process is one dimensional. In Section 4 we apply the general results to the problem with probability distortion.
We demonstrate how time-inconsistency may arise in this case, followed by deriving
sophisticated equilibrium strategies in the presence of time inconsistency.
Specifically, Subsection~\ref{subsec:demonstration} presents a detailed case study of how one can find a large class of equilibria and possibly compare/rank them, under a particular model specification; Subsection~\ref{subsec:examples} contains a series of practical examples where the payoff is either a power utility or a European option payoff function, and obtains both na\"ive and sophisticated strategies. Finally,  Appendices collect technical proofs.


\section{Na\"ive and Equilibrium Stopping Laws}\label{sec:equilibrium}

This section sets out to introduce time-inconsistency for a general optimal stopping problem and
the ways we deal with the issue.

Consider a probability space $(\Omega,\Fc,\P)$ 
that supports a strong Markov 
process $X:\R_+\times\Omega\mapsto\X$ with initial value $X_0 =x\in\X$, where $\X$, the state space of $X$, is a connected subset of $\R^d$. We will constantly use the notation $X^x$ to emphasize the dependence of the process on the initial value $x$. We denote by $\Bc(\X)$ the Borel $\sigma$-algebra of $\X$. Let $\F=\{\Fc_t\}_{t\ge 0}$ be the $\P$-augmentation of the natural filtration generated by $X$, and let $\Tc$ be the set of all $\F$-stopping times that may be infinite with positive probabilities. Note that $\Tc$ includes non-Markovian stopping times\footnote{Non-Markovian stopping
is widely used in practice. For example, the ``trailing stop" in stock trading is a non-Markovian selling rule. \cite{HHOZ17} shows that in a casino gambling model under CPT path-dependent stopping strategies
strictly dominate Markovian ones.}.

At any given state $x\in\X$, a decision-maker (an {\it agent}) needs to decide when to stop the process $X$. 
If he chooses to stop at $\tau\in\Tc$, he receives the payoff $J(x;\tau)$, where $J:\X\times\Tc\mapsto\R$ is a given objective functional.
The agent intends to maximize his payoff by choosing an appropriate stopping time $\tau\in\Tc$, i.e., he aspires to achieve
\begin{equation}\label{v}
\sup_{\tau\in\Tc} J(x;\tau).
\end{equation}
A well-studied example of the objective functional is the expected payoff
\begin{equation}\label{standard J}
J(x;\tau) := \E[u(X^x_\tau)],
\end{equation}
where $u:\X\mapsto\R$, a
Borel measurable function, is the so-called payoff function; namely, $u(x)$ is the payoff he receives if he decides to stop
immediately.  
Another example is a problem with  general non-exponential discounting:
\begin{equation}\label{non-exponential}
J(x;\tau) := \E[\delta(\tau)u(X^x_\tau)],\quad 
\end{equation}
where $\delta:[0,\infty)\to [0,1]$ is decreasing with $\delta(0)=1$. The mean-variance criterion can also be considered, with
\begin{equation}\label{M-V}
J(x;\tau) := \E[X^x_\tau] - \gamma \text{Var}(X_\tau^x),
\end{equation}
for some $\gamma>0$. In Section~\ref{sec:prob dist} we will  focus on an objective functional where the probability is ``distorted", as defined in \eqref{J} below.

As $\{X^x_t\}_{t\ge 0}$ evolves over time, one can reexamine and solve the optimal stopping problem \eqref{v} at every moment $t\ge 0$. A natural, conceptual question arises: suppose an optimal stopping time $\hat\tau_x$ can be found for each and every $x\in\X$, are $\hat\tau_x$ and $\hat{\tau}_{X^x_t}$, $t>0$, consistent with each other?

The notion of ``consistency" in this question can be formalized as follows:
\begin{Definition}[Time-Consistency]\label{def:time-consistency}
Suppose an optimal stopping time $\hat{\tau}_x\in\Tc$ of the problem \eqref{v} exists for all $x\in\X$. We say problem \eqref{v} is time-consistent if for any $t>0$
and any $x\in\X$,
\begin{equation}\label{tidef}
\hat{\tau}_x = t+\hat{\tau}_{X^x_t}\quad \hbox{a.s. on}\ \ \{\hat\tau_x\ge t\};
\end{equation}
otherwise problem \eqref{v} is time-inconsistent.
\end{Definition}
Intuitively, time-inconsistency means that an optimal strategy planned at this moment may no longer be optimal at the next moment. It is well known that problem \eqref{v} with the expected payoff  in \eqref{standard J} is time-consistent.\footnote{With the expected payoff \eqref{standard J}, the tower rule holds:
$\E[u(X^x_\tau)]=\E\left[\E[u(X^x_\tau)|{\cal F}_t]\right]$
for any
$\tau\in\Tc$ and $t\geq0$. This means that the stopping problem at a given time $t$
is consistent with the stopping problem at the initial time, so long as the problem has not yet been stopped  at $t$. Hence one can apply classical
approaches such as dynamic programming and martingale method to solve it; see e.g., \cite{S-book-78} and \cite{KS-book-98}.}
For a general objective functional, problem \eqref{v} is predominantly time-inconsistent - for example when probability distortion is involved or non-exponential discounting is applied.

Time-inconsistency renders the very notion of ``dynamic optimization" null and void, for if one cannot commit his future selves to the optimal strategy he chooses today, then
today's optimal strategy has little use in a dynamic context. More specifically, for any given  state $x\in\X$ at $t=0$, suppose the agent finds an optimal stopping time $\hat\tau_x$. He actually wants, and indeed
 assumes, that all his future selves {\it will} follow $\hat\tau_x$, so that the optimal value $\sup_{\tau\in\Tc}J(x;\tau)=J(x;\hat\tau_x)$ can be attained. However, his future self at time $t>0$ would like to follow his {\it own} optimal stopping time $\hat\tau_{X^x_t}$, which may not be consistent with $\hat\tau_x$
 in the sense of (\ref{tidef}). If the agent at time $0$ does not have sufficient control over his future selves' behavior, $\hat\tau_x$ will not be carried  through, and the optimal value $\sup_{\tau\in\Tc}J(x;\tau)$ as initially figured out at $t=0$ will thus {\it not} be attained.

As discussed in the introduction,
three types of agents have been described  in literature in the presence of time-inconsistency.
A na\"{i}ve agent simply follows an optimal stopping time $\hat\tau_{X^x_t}$ at every moment $t\ge 0$, without knowing  the underlying time-inconsistency.
A sophisticated agent who is aware of the time-inconsistency but lacks commitment, by contrast, works on consistent planning:
he takes into account his future selves' behaviors,
and tries to find a stopping strategy that, once being employed over time,
none of his future selves has incentives to deviate from it.
Finally, a sophisticated agent who is able to commit simply solves the problem {\it once} at $t=0$ and then sticks to the corresponding stopping plan.
The problem of the last type, the so-called pre-committed agent,
is actually a static (instead of dynamic) problem and has been solved in various contexts, such as optimal stopping under probability distortion in \cite{XZ13}, mean--variance portfolio selection in \cite{ZL00}, optimal stopping with nonlinear constraint on expected time to stop in \cite{Miller06}, and optimal control of conditional Value-at-Risk in \cite{MY17}.
The goal of this paper is to study the behaviors of the first two types of agents.
It turns out, as we will demonstrate, that the solutions of the two are interwound,
and both depend heavily on that of the last type.

Now we provide precise formulations of the stopping strategies of the first two types of agents (hereafter referred respectively
as the na\"{i}ve and sophisticated agents). We first introduce the notion of {\it stopping laws}, taken from Definition 3.1 in \cite{HN16}.

\begin{Definition}[Stopping Laws]\label{def:stopping policy}
A Borel measurable function $\tau:\X\mapsto\{0,1\}$ is called a (Markovian) stopping law. We denote by $\Tc(\X)$ the collection of all stopping laws.
\end{Definition}

The notion of stopping laws is analogous to that of feedback control laws in control theory. A stopping law is {\it independent} of any state process; however for any given process, such a law  induces
a stopping decision (in response to any current state) in the following manner.

Given the process $X$ with $X_0=x$, each stopping law $\tau\in\Tc(\X)$ governs when the agent stops $X$: the agent stops at the first time when $\tau(X_t)$ yields the value $0$, i.e., at the moment
\begin{equation}\label{L}
\Lc\tau(x):= \inf\{t\ge 0: \tau(X^x_t) = 0\}.
\end{equation}
In other words, $\Lc\tau(x)$ is the stopping time {\it induced} by the stopping law $\tau$ when the current state of the process is $x$.

We first define the stopping law used by a na\"{i}ve agent.

\begin{Definition}[Na\"{i}ve Stopping Law]\label{def:induced policy}
Denote by $\{\hat\tau_x\}_{x\in\X}\subset\Tc$ a collection of optimal stopping times of \eqref{v}, while noting that $\hat\tau_x$ may not exist for some $x\in\X$. Define $\hat\tau:\X\mapsto\{0,1\}$ by
\begin{equation}\label{induced policy}
\hat\tau(x):=
\begin{cases}
0,\quad \hbox{if}\ \hat\tau_{x}=0,\\
1,\quad \hbox{if}\ \hat\tau_{x}>0\ \hbox{or}\ \hat\tau_x\ \hbox{does not exist}.
\end{cases}
\end{equation}
If $\hat\tau$ is Borel measurable, we say that it is the {\it na\"{i}ve stopping law} generated by $\{\hat\tau_x\}_{x\in\X}$.
\end{Definition}

The stopping law $\hat\tau\in\Tc(\X)$ defined above describes the behavior of a na\"{i}ve agent. For any current state $x>0$, a na\"{i}ve agent decides to stop or to continue simply by following the optimal stopping time $\hat\tau_x$, if such a stopping time exists. If $\hat\tau_x$ fails to exist at some $x\in\X$, we must have $\sup_{\tau\in\Tc}J(x;\tau)> u(x)$ (otherwise, we can take $\hat\tau_x = 0$). Although the optimal value $\sup_{\tau\in\Tc}J(x;\tau)$ cannot be attained, the na\"{i}ve agent can pick some $\tau\in\Tc$ with $\tau>0$ a.s. such that $J(x;\tau)$ is arbitrarily close to $\sup_{\tau\in\Tc}J(x;\tau)$, leading to $J(x;\tau)> u(x)$. That is, a na\"{i}ve agent determines that it is better to continue than to stop at $x\in\X$. This is why we set $\hat\tau(x)=1$ when $\hat\tau_x$ does not exist in the above definition\footnote{There is a technical question as to whether optimal stopping times $\{\hat\tau_x\}_{x\in\X}$ of \eqref{v} can be chosen such that $\hat\tau$ defined in \eqref{induced policy} is Borel measurable. The answer is positive under the standard expected payoff formulation \eqref{standard J}. For an objective functional under probability distortion, we will see in Sections~\ref{sec:prob dist} that $\hat\tau$ is measurable in all examples we explore.}.

The formulation of na\"{i}ve stopping laws first appeared in Remark 3.3 of \cite{HN16} for a concrete problem under non-exponential discounting. Here, we extend the notion to a general setting  and allow the possibility that an optimal stopping time does not exist.

\begin{Remark}
Time-consistency in Definition~\ref{def:time-consistency} can also be formulated as
\begin{equation}\label{tidef'}
\hat\tau_x = \Lc\hat\tau(x)= \inf\{t\ge 0: \hat\tau_{X^x_t}=0\}\quad \hbox{a.s.}\quad \forall x\in\X.
\end{equation}
Note that the second equality follows from \eqref{induced policy} directly. The equivalence between \eqref{tidef} and \eqref{tidef'} can be observed as follows. Suppose \eqref{tidef} holds. For any $x\in\X$ and a.e. $\omega\in\Omega$, let $\ell:= \hat\tau_x(\omega)\ge 0$. By \eqref{tidef}, $\hat\tau_{X^x_t}(\omega) = \ell-t$, $\forall\ 0\le t\le \ell$. This yields $\Lc\hat\tau(x)(\omega) = \ell = \hat\tau_x(\omega)$, and thus \eqref{tidef'} holds. Conversely, suppose \eqref{tidef'} is true. For any $x\in\X$ and $t>0$, on the set $\{\hat\tau_x\ge t\}$ we have $\inf\{s\ge 0: \hat\tau_{X^x_s}=0\} = \Lc\hat\tau(x) = \hat\tau_x \ge t$, and thus $\hat\tau_{X^x_s}\neq 0$ for $0\le s<t$. This implies $\inf\{s\ge 0: \hat\tau_{X^x_s}=0\} = t+\inf\{s\ge 0: \hat\tau_{X^y_s}=0\}$, with $y:= X^x_t$. This, together with \eqref{tidef'}, gives $\hat\tau_x = t+\hat\tau_{X^x_t}$ on $\{\hat\tau_x\ge t\}$, i.e., \eqref{tidef} is proved.

The formulation \eqref{tidef'} means that the moment at which a na\"{i}ve agent will stop, i.e., $\Lc\hat\tau(x)$, is entirely the same as what the original pre-committed optimal stopping time (which was planned at time 0 when the process was in state $x$), $\hat\tau_x$, prescribes.
\end{Remark}

Now we turn to characterizing a sophisticated agent, by using the game-theoretic argument introduced in Section 3.1 of \cite{HN16}.
Suppose the agent starts with an initial stopping law $\tau\in\Tc(\X)$. At any current state $x\in\X$, the agent carries out the following game-theoretic reasoning: if all my future selves will follow $\tau\in\Tc(\X)$, what is the best stopping strategy at current time in response to that? As the agent at current time can only choose to stop or to continue, he simply needs to compare the payoffs resulting from these two different actions. If the agent stops at current time, he obtains the payoff $u(x)$ immediately. If the agent chooses to continue at current time, {\it given} that all his future selves will follow $\tau\in\Tc(\X)$, the agent will eventually stop at the moment
\begin{equation}\label{L*}
\Lc^*\tau(x):= \inf\{t> 0: \tau(X^x_t) = 0\},
\end{equation}
leading to the payoff
\[
J(x;\Lc^*\tau(x)).
\]
Note the subtle difference between the two stopping {\it times}, $\Lc^*\tau(x)$ and $\Lc\tau(x)$: at any moment one may continue under the former even though the latter may instruct to stop.

Some conclusions can then be drawn: (i) The agent should stop at current time if $u(x)>J(x;\Lc^*\tau(x))$, and continue at current time if $u(x)<J(x;\Lc^*\tau(x))$. (ii) For the case where $u(x)= J(x;\Lc^*\tau(x))$, the agent is indifferent between stopping and continuation at current time; there is therefore no incentive for the agent to deviate from the originally assigned stopping strategy $\tau(x)$ now. This is already the best stopping strategy (or law) at current time (subject to all the future selves following $\tau\in\Tc(\X)$), which can be summarized as
\begin{equation}\label{Theta}
\Theta \tau (x) :=
\begin{cases}
0,&\quad\hbox{if}\ x\in S_\tau,\\
1,&\quad\hbox{if}\ x\in C_\tau,\\
\tau(x),&\quad\hbox{if}\ x\in I_\tau,
\end{cases}
\end{equation}
where
\begin{align*}
S_\tau &:= \{x\in\X:J(x;\Lc^*\tau(x))<u(x)\},\\
C_\tau &:= \{x\in\X:J(x;\Lc^*\tau(x))>u(x)\},\\
I_\tau &:= \{x\in\X:J(x;\Lc^*\tau(x))=u(x)\}
\end{align*}
are the stopping region, the continuation region, and the indifference region, respectively.

\begin{Remark}
The definitions of $\Theta$ in \eqref{Theta} and the three regions $S_\tau$, $C_\tau$, and $I_\tau$ are motivated by the formulation in \cite{HN16}. There, the objective functional $J$ is restricted to the form \eqref{non-exponential}. This paper allows for a much larger class of objective functionals, specified by Assumption~\ref{asm:J} below.
\end{Remark}

\begin{Remark}[Non-Existence of Optimal Stopping Times]
The way we formulate time-consistency in Definition~\ref{def:time-consistency}, which follows a long thread of literature in Economics and Mathematical Finance, hinges on the existence of (pre-committed) optimal controls/stopping times at every state. When an optimal strategy fails to exist, it is unclear how to define time-consistency. Recently, \cite*{KMZ17} point out this problem, and propose a possible way to formulate time-consistency, without referring to optimal strategies, via backward stochastic differential equations (BSDEs).
However,  our game-theoretic approach described above -- which eventually leads to a sophisticated stopping strategy -- does not rely on the existence of optimal strategies.
\end{Remark}

Given any arbitrarily given  stopping law $\tau\in\Tc(\X)$, the above game-theoretic thinking gives rise to an alternative stopping law, $ \Theta \tau$, which is {\it at least as good as} $\tau$ to this
sophisticated agent. Naturally, an equilibrium stopping law can be defined as one that is invariant under  such a  game-theoretic reasoning. This motivates Definition \ref{def:equilibrium} below.

However, to carry this idea through, we need to identify conditions for the objective functional $J$ under which $ \Theta \tau$ is indeed a stopping law satisfying the measurability requirement per Definition \ref{def:stopping policy}.
To this end, for any $\tau:\X\mapsto\{0,1\}$, we consider the {\it kernel} of $\tau$, which is the collection of states
at which the agent stops,  defined by
\[
\ker(\tau) := \{x\in\X:\tau(x)=0\}.
\]

\begin{Remark}
For any $\tau\in\Tc(\X)$ and $x\in\X$, $\Lc\tau(x)$ and $\Lc^*\tau(x)$ belong to $\Tc$. Indeed, the measurability of $\tau\in\Tc(\X)$ implies
\begin{equation}\label{kernel Borel}
\ker(\tau) \in \Bc(\X).
\end{equation}
Thanks to the right continuity of the filtration $\F$,
\begin{equation}\label{hit ker}
\Lc\tau(x) = \inf\{t\ge 0: X^x_t\in\ker(\tau)\}\quad \hbox{and}\quad \Lc^*\tau(x) = \inf\{t> 0: X^x_t\in\ker(\tau)\}
\end{equation}
are $\F$-stopping times.
\end{Remark}

Now we introduce the assumptions that will ensure the measurability of $\Theta\tau$.

\begin{Assumption}\label{asm:J}
The objective function $J:(\X,\Tc)\mapsto\R$ satisfies
\begin{itemize}
\item [(i)] for any $D\in\Bc(\X)$, the map $x\mapsto J(x;T^x_D)$ is Borel measurable, where
\begin{equation}\label{T_D}
T^x_D:=\inf\{t>0:X^x_{t}\in D\}.
\end{equation}
\item [(ii)] for any sequence $\{D_n\}_{n\in\N}$ in $\Bc(\X)$ such that $D_n\subseteq D_{n+1}$ for all $n\in\N$,
\[
\liminf_{n\to\infty}J(x;T^x_{D_n})\le  J(x;T^x_D),\quad \hbox{where}\ D:= \cup_{n\in\N} D_n.
\]
\end{itemize}
\end{Assumption}

\begin{Remark}\label{rem:satisfied}
Assumption~\ref{asm:J} is very mild, and is satisfied by the classical formulation \eqref{standard J}, as well as stopping under non-exponential discounting \eqref{non-exponential} and the mean-variance criterion \eqref{M-V}.
As part (i) is easy to verify in these applications, the discussion below focuses on part (ii).

Under \eqref{standard J}, to ensure that problem \eqref{v} is well-defined, a standard condition imposed is $\E\big[\sup_{0\le t\le \infty}u(X^x_t)\big] <\infty$, where $u(X^x_\infty):= \limsup_{t\to\infty} u(X^x_t)$. The dominated convergence theorem can then be applied to get
\begin{equation}\label{J to J}
\lim_{n\to\infty}J(x;T^x_{D_n})=  J(x;T^x_D)
\end{equation}
for $\{D_n\}\subseteq\Bc(\X)$, $D_n\subseteq D_{n+1}$, and $D:= \cup_{n\in\N} D_n$. 
Similarly, under \eqref{non-exponential} (resp. \eqref{M-V}), to ensure that \eqref{v} is well-defined, a standard condition imposed is $\E\big[\sup_{0\le t\le \infty}\delta(t)u(X^x_t)\big] <\infty$ (resp. $\E\big[\sup_{0\le t\le \infty}(X^x_t)^2\big] <\infty$). The dominated convergence theorem can then be applied to get \eqref{J to J}. With \eqref{J to J}, Assumption~\ref{asm:J}-(ii) is trivially satisfied.
\end{Remark}

\begin{Proposition}
Suppose Assumption~\ref{asm:J} (i) holds. Then $\Theta\tau\in\Tc(\X)$ whenever $\tau\in\Tc(\X)$.
\end{Proposition}

\begin{proof}
In view of \eqref{hit ker}, $\Lc^*\tau(x)$ is simply $T^x_{\ker(\tau)}$. Thus, by Assumption~\ref{asm:J} (i), $x\mapsto J(x;\Lc^*\tau(x))$ is Borel measurable, whence $S_\tau$, $C_\tau$, and $I_\tau$ are all in $\Bc(\X)$. By \eqref{Theta} and \eqref{kernel Borel},
\begin{align*}
\ker(\Theta\tau)&= S_\tau \cup (I_\tau\cap \ker(\tau))\in\Bc(\X),
\end{align*}
which implies that $\Theta\tau:\X\mapsto \{0,1\}$ is Borel measurable.
\end{proof}

\begin{Definition}[Equilibrium (Sophisticated) Stopping Law]\label{def:equilibrium}
A stopping law $\tau\in\Tc(\X)$ is called an equilibrium if $\Theta\tau(x)=\tau(x)$ for all $x\in\X$. We denote by $\Ec(\X)$ the collection of all equilibrium stopping laws.
\end{Definition}

\begin{Remark}[Trivial Equilibrium]\label{rem:trivial}
As with almost all the Nash-type equilibria, existence and uniqueness are important problems.
 In our setting, a stopping law $\tau\in\Tc(\X)$ defined by $\tau(x)=0$ for all $x\in\X$ is trivially an equilibrium. Indeed, for any $x\in\X$, $\Lc^*\tau(x)=0$ and thus $J(x;\Lc^*\tau(x))=u(x)$. This implies $I_\tau = \X$. By \eqref{Theta}, we conclude $\Theta\tau(x)=\tau(x)$ for all $x\in\X$.
\end{Remark}

To search for equilibrium stopping laws, the general (and natural) idea is to perform fixed-point iterations on the operator $\Theta$:
starting from any $\tau\in\Tc(\X)$, take
\begin{equation}\label{tau0}
\tau_*(x) := \lim_{n\to\infty} \Theta^n\tau(x)\quad x\in\X.
\end{equation}

The above procedure admits a clear economic interpretation. At first, the (sophisticated) agent has an initial
stopping law $\tau$.  Once he starts to carry out the game-theoretic reasoning stipulated earlier, he realizes that the best stopping strategy for him, given that all future selves will follow $\tau$, is $\Theta\tau$. He therefore switches from $\tau$ to $\Theta\tau$. The same game-theoretic reasoning then implies that the best stopping strategy for him, given that all future selves will follow $\Theta\tau$, is $\Theta^2\tau$. The agent thus switches again from $\Theta\tau$ to $\Theta^2\tau$. This procedure continues until the agent eventually arrives at an equilibrium $\tau_*$, a fixed point of $\Theta$ and a strategy he cannot further improve upon by the procedure just described. In economic terms, each application of $\Theta$ corresponds to an additional level of strategic reasoning.

Mathematically, we need to prove that the limit taken in \eqref{tau0} is well-defined, belongs to $\Tc(\X)$,   and satisfies $\Theta\tau_* = \tau_*$. In general, such results are not easy to establish, and remain largely an open question\footnote{For a stopping problem with expected utility and non-exponential discounting, \cite{HN16} show that fixed-point iterations do converge to equilibria, under appropriate conditions on the discount function and the initial stopping law $\tau\in\Tc(\X)$; see Theorem 3.1 and Corollary 3.1 therein.}.

 However, when $X$ is a one-dimensional diffusion process we will be able to derive the stopping strategies
 of the na\"{i}ve 
 and the sophisticated agents in a fairly complete and explicit manner.
 This is not only because, in the one-dimensional case,
 pre-committed stopping laws have been obtained rather thoroughly by \cite{XZ13} on which a na\"{i}ve strategy depends,
 but also because the fixed-point iteration \eqref{tau0} turns out to be much more manageable and does converge to an equilibrium,
 due to a key technical result (Lemma \ref{lem:T_x=0}) that holds only for a one-dimensional process.

	
\section{Equilibrium Stopping Laws: The One-Dimensional Case}\label{subsec:1-D case}
Let $\X$ be an open interval in $\R$, i.e., $\X=(\ell,r)$ for some $-\infty\le \ell<r\le\infty$. Consider Borel measurable $a,b:\X\mapsto \R$, with $a(\cdot)>0$, such that for any $x\in\X$, the stochastic differential equation 
\[
dX_t = b(X_t) dt + a(X_t) d B_t,\quad X_0=x,
\]
where $B$ denotes a standard one-dimensional Brownian motion, has a weak solution that is unique in law and does not attain the boundaries of $\X$ a.s., i.e., $\P(X^x_t \in (\ell,r)\ \forall t\ge 0)=1$.\footnote{This is guaranteed by appropriate conditions on $a$ and $b$. Specifically, $a(\cdot)>0$ and local integrability of $|b(\cdot)|/a^2(\cdot)$ imply the existence and uniqueness in law of a weak solution $X$. On the other hand, whether $X$ attains the boundaries of $\X$ is fully characterized by Feller's test for explosions. We refer to Sections 5.5.B, 5.5.C of \cite{KS-book-91} for a detailed exposition.} 

Define $\theta(\cdot):=b(\cdot)/a(\cdot)$,
and assume that
\begin{equation}\label{integrability}
\P\left[\int_0^t \theta^2(X_s) ds <\infty \right] =1,\quad \forall t\ge 0.
\end{equation}
Introduce the process
\begin{equation}\label{Z}
Z_t := \exp\left(-\int_0^t \theta(X_s)dB_s - \frac12 \int_0^t\theta^2(X_s)ds \right)\quad \hbox{for}\ t\ge 0,
\end{equation}
which is by definition a nonnegative local martingale.
Note that \eqref{integrability} is a standard condition ensuring that $\int_0^t \theta(X_s)dB_s$ is well-defined for all $t\ge 0$ a.s.

\begin{Proposition}\label{prop:ker increases}
Suppose $\X$ is an open interval in $\R$, 
Assumption~\ref{asm:J} holds, and $Z$ defined in \eqref{Z} is a martingale.
Then, for any $\tau\in\Tc(\X)$,
\begin{align}\label{ker increases}
&\ker(\tau)\subseteq I_\tau,\nonumber\\
\ker(\Theta^{n}\tau) =\ &S_{\Theta^{n-1}\tau}\cup\ker(\Theta^{n-1}\tau),\quad\forall n\in\N.
\end{align}
Hence, $\tau_*$ in \eqref{tau0} is well-defined, with
$
\ker(\tau_*) = \bigcup_{n\in\N} \ker(\Theta^n\tau).
$
\end{Proposition}

The proof of Proposition~\ref{prop:ker increases} is relegated to Appendix~\ref{appen:T_x=0}.

\begin{Theorem}\label{thm:main}
Suppose $\X$ is an open interval in $\R$, 
Assumption~\ref{asm:J} holds, and $Z$ defined in \eqref{Z} is a martingale.
Then, for any $\tau\in\Tc(\X)$, $\tau_*$ defined in \eqref{tau0} belongs to $\Ec(\X)$.  Hence,
\begin{equation}\label{Ec characterization}
\Ec(\X) =
\{\tau_*\in\Tc(\X) ~:~  \tau_* = \lim_{n\to\infty} \Theta^n\tau,\ \hbox{for some}\ \tau\in\Tc(\X)\} = \bigcup_{\tau\in\Tc(\X)}\{\lim_{n\to\infty} \Theta^n\tau\} .
\end{equation}
\end{Theorem}


\begin{proof}
By Proposition~\ref{prop:ker increases},  $\tau_*\in\Tc(\X)$ is well-defined and $\ker(\tau_*) = \bigcup_{n\in\N} \ker(\Theta^n\tau)$. As $\ker(\tau_*)\subseteq I_{\tau_*}$ (by Proposition~\ref{prop:ker increases} again), $\Theta\tau_*(x) = \tau_*(x)$ for all $x\in\ker(\tau_*)$. For $x\notin\ker(\tau_*)$, we have $x\notin \ker(\Theta^n\tau)$, i.e., $\Theta^n\tau(x)=1$, for all $n\in\N$. In view of \eqref{Theta}, this gives $J(x;\Lc^*\Theta^{n-1}\tau(x))\ge u(x)$ for all $n\in\N$.
By \eqref{hit ker}, this can be written as $J(x;T^x_{\ker(\Theta^{n-1}\tau)})\ge u(x)$ for all $n\in\N$. Thanks to Proposition~\ref{prop:ker increases}, $\{\ker(\Theta^{n-1}\tau)\}_{n\in\N}$ is a nondecreasing sequence of Borel sets and $\ker(\tau_*) = \bigcup_{n\in\N}\ker(\Theta^{n-1}\tau)$. It then follows from Assumption~\ref{asm:J} (ii) that
\[
J(x;T^x_{\ker(\tau_*)})\ge \liminf_{n\to\infty} J(x;T^x_{\ker(\Theta^{n-1}\tau)})\ge u(x).
\]
This implies $x\in C_{\tau_*}\cup I_{\tau_*}$. If $x\in I_{\tau_*}$, then $\Theta\tau_*(x)=\tau_*(x)$; if $x\in C_{\tau_*}$, then $\Theta\tau_*(x)=1=\tau_*(x)$, as $x\notin\ker(\tau_*)$. Thus, we conclude that $\Theta\tau_*(x)=\tau_*(x)$ for all $x\in\X$, i.e., $\tau_*\in\Ec(\X)$.

For the first equality in \eqref{Ec characterization}, the ``$\supseteq$'' relation holds as a direct consequence of the above result. Note that the ``$\subseteq$'' relation is also true because for each $\tau_*\in \Ec(\X)$, one can simply take $\tau\in\Tc(\X)$ on the right hand side to be $\tau_*$ itself. The last set in \eqref{Ec characterization} is simply a re-formulation of the second set.
\end{proof}

This theorem indicates that under its assumptions every equilibrium can be found through the fixed-point iteration. Moreover, it stipulates a way of telling whether
a given stopping law is an equilibrium. Any initial strategy that can be {\it strictly}
improved by the iteration (namely the game-theoretic reasoning) is {\it not} an equilibrium strategy.
On the other hand, even if
the agent's initial strategy happens to be already an equilibrium, he may not realize it. He does only after he has applied the iteration and found that he gets the same strategy.

Any given stopping law $\tau$ will give rise to an equilibrium $\tau_*$ according to (\ref{Ec characterization}); so in general it is unlikely that we will have
a unique equilibrium.
In this paper, we are particularly interested in the equilibrium $\hat\tau_*\in\Ec(\X)$ generated by the na\"{i}ve stopping law $\hat\tau\in\Tc(\X)$ which is induced by  $\{\hat\tau_x\}_{x\in\X}$ in \eqref{induced policy}; that is,
\begin{equation}\label{wtau*}
\hat\tau_*(x) = \lim_{n\to\infty} \Theta^n\hat \tau(x)\quad x\in\X.
\end{equation}
The economic significance of choosing such an equilibrium is that it spells out explicitly how an (irrational) na\"{i}ve agent might be turned into a (fully rational) sophisticated one, if he is educated to carry out sufficient levels of strategic reasoning.


\section{Application to Problems with Probability Distortion}
\label{sec:prob dist}
In this section we apply the general setting and results established in the previous sections to the case where the objective functional involves probability distortion (weighting) and the underlying process is a geometric Brownian motion. 

Let $S=\{S_t\}_{t\ge 0}$ be a geometric Brownian motion satisfying
\begin{equation}
\label{eq:P dynamics}
dS_t = \mu S_t dt + \sigma S_t dB_t, \quad S_0=s >0,\quad \hbox{for}\ t\ge 0,
\end{equation}
where $\mu\in\R$ and $\sigma>0$ are constants and $\{B_t\}_{t\ge 0}$ is a standard Brownian motion. Note that the process $S$ takes values in $\R_+ :=(0,\infty)$ almost surely. 
In most of the discussions in this paper, $S$ is interpreted as the price process of an asset, although it could represent other processes such as the value process of a project and similar  discussions could be conducted.

Consider a nondecreasing, continuous function $U:\R_+\mapsto \R_+$ with $U(0)=0$, and a strictly increasing, continuous function $w: [0,1]\mapsto [0,1]$ with $w(0)=0$ and $w(1)=1$. Here, $U$ is the utility function of the agent or the payoff function of the asset, and $w$ the probability distortion function.
The agent's objective is, for any given initial state $s$, to maximize his ``distorted'' expected payoff
\begin{equation}
\label{I}
\hbox{maximize}\quad I(s; \tau):=\int_0^\infty w\left(\Pro{U(S^s_\tau)>y} \right)dy,
\end{equation}
 by choosing an appropriate stopping time $\tau \in \Tc$.\footnote{In the literature,
 the distorted expected payoff is formulated as $I(s; \tau)=\int_0^\infty U(y)d[1-w(1-F(y))]$, where $F$ is the distribution function of $S^s_\tau$.
 This expression and \eqref{I} are equivalent by Fubini's theorem. Indeed,
 $\int_0^\infty U(y)d[1-w(1-F(y))]=\int_0^\infty \int_0^{U(y)}dxd[1-w(1-F(y))]
 =\int_0^\infty \int_{U^{-1}(x)}^\infty d[1-w(1-F(y))]dx=\int_0^\infty w\left(1-F(U^{-1}(x))\right)dx=\int_0^\infty w\left(\Pro{U(S^s_\tau)>x} \right)dx$,
 where $U^{-1}$ is the left (or right) inverse of $U$.
 %
}

So the agent intends to maximize his expected utility (or payoff) under probability distortion, by stopping the price process at an appropriate moment. This formulation is motivated by several financial applications, such as liquidation of a financial position, real options  and casino gambling. 
Note that with $w(p)= p$, we retrieve the standard expectation valuation.

Throughout the analysis in this section, we consider the parameter
\begin{equation}\label{beta}
\beta:=1-2\mu/\sigma^2.
\end{equation}
The value of $\beta$ and how ``good'' the asset $S$ is are inter-connected. \cite*{SXZ08} take $\frac{\mu}{\sigma^2}$ as the ``goodness index'' of an asset, which is included in \eqref{beta}. The larger this index is, the higher the asset's expected return relative to its volatility.



In the rest of this section, we will focus on the case $\beta>0$, where time-inconsistency arises genuinely for \eqref{I}.\footnote{For the case $\beta\le 0$, one may deduce from \cite{XZ13} that problem \eqref{I} is time-consistent. Indeed, when $\beta=0$, \cite{XZ13}, p. 255, proves that
$\hat\tau_{s} := \inf\{t\ge 0: S^s_t \ge x^*\} = \inf\{t\ge 0: B_t \ge \sigma^{-1} \log(x^*/s)\}$ is a (pre-committed) optimal stopping time of \eqref{I}, where $x^*:= \inf\left\{x>0:U(x) = \sup_{y>0}U(y)\right\}$. The stopping threshold for the state process $S^s$ is always $x^*$, independent of the current state $s$. There is then no time-inconsistency. When $\beta<0$, Theorem 2.1 in \cite{XZ13} shows that $\hat\tau_s \equiv \infty$, implying ``never-stopping",  is the (pre-committed) optimal stopping time. Again, no time-inconsistency is present here as $\hat\tau_s \equiv \infty$ does not depend on current state $s$. }

\subsection{The Case $\beta>0$: Time-Inconsistency}
We consider, as in Section 2.2 of \cite{XZ13}, the process $X_t := S_t^\beta$, $t\ge 0$. That is, $X$ is the Dol\'eans-Dade exponential of $\{\beta\sigma B_t\}_{t\ge 0}$:
\begin{equation}
\label{eq:X}
dX^x_t = \beta \sigma X^x_t dB_t, \quad X^x_0 = x>0,
\end{equation}
where $x = s^\beta>0$. 
As a geometric Brownian motion with no drift, $X^x$ is a $\P$-martingale and
\begin{equation}\label{0 at infinity}
\lim_{t\to\infty} X^x_t = 0\quad \hbox{a.s.},
\end{equation}
thanks to the law of iterated logarithms. As we will see subsequently, this fact plays an important role
in dictating a sophisticated agent's behavior - the value of the underlying process
diminishes in the long run. Moreover, for any $b>x$,
\begin{equation}\label{Tb<infinity}
\P[T^x_b<\infty] = \frac xb <1,\quad \hbox{where}\quad T^x_b := \inf\{t>0 : X^x_t=b\}.
\end{equation}
Note that \eqref{Tb<infinity} follows from the hitting time probability of a standard Brownian motion to a linear boundary; see e.g., \cite{KS-book-91}.

Let $u$ be defined by
\begin{equation}\label{u}
u(x):=U(x^{1/\beta}),\quad x\in(0,\infty).
\end{equation}
With $\beta>0$, the transformed function $u$ is nondecreasing with $u(0)= U(0)=0$. 
The shape of $u$, which can be convex, concave, or $S$-shaped (i.e., first convex, and then concave), depends on both the shape of $U$ and the coefficient $\beta$. Similarly, we define
\begin{equation}
\label{J}
J(x; \tau) :=I(x^{1/\beta};\tau)\equiv \int_0^\infty w\left(\Pro{u(X_\tau^x)>y} \right)dy\ \quad \hbox{for}\ x>0,\ \tau \in \Tc.
\end{equation}
Here, we allow $\tau\in\Tc$ to take the value $+\infty$ with positive probability: on the set $\{\tau=\infty\}$, we simply take $X_\tau=0$ in view of \eqref{0 at infinity}.


Under the objective functional \eqref{J}, \cite{XZ13} characterize the pre-committed optimal stopping times for problem \eqref{v}, stating that the problem may be time-inconsistent as $\beta>0$. The next example demonstrates this time-inconsistency explicitly.

\begin{Example}\label{eg:never stop}
Take $u(x)=x^\eta$ with $\eta\ge 1$, and consider the probability weighting function proposed by \cite{Prelec98}:
\begin{equation}\label{Prelec's w}
w(x)=\exp(-\gamma(-\log x)^\alpha)\quad \hbox{for some}\ \alpha,\gamma>0.
\end{equation}
As $u$ is convex, Theorem 4.5 of \cite{XZ13} shows that problem \eqref{v}, with $J(x;\tau)$ as in \eqref{J}, can be reduced to the optimization problem
\begin{equation}\label{v reduced}
\sup_{\lambda\in(0,1]} w(\lambda)u\left(\frac{x}{\lambda}\right).
\end{equation}
To solve this problem for our case, for each $x>0$, set $f(\lambda):=w(\lambda)u\left(\frac{x}{\lambda}\right)$. 
By direct computation,
\begin{align*}
f'(\lambda) &= \frac{w(\lambda)}{\lambda}\left(\frac{x}{\lambda}\right)^\eta[\alpha\gamma(-\log\lambda)^{\alpha-1}-\eta].
\end{align*}
Observe that $f'(\lambda)=0$ has a {\it unique} solution $\lambda^* = e^{-\left(\frac{\eta}{\alpha\gamma}\right)^{1/(\alpha-1)}}$ on $(0,1]$. Moreover,
\[
f''(\lambda^*) = -\alpha\gamma(\alpha-1)\left(\frac{\eta}{\alpha\gamma}\right)^{\frac{\alpha-2}{\alpha-1}}\frac{w(\lambda^*)}{(\lambda^*)^2 } \left(\frac{x}{\lambda^*}\right)^\eta.
\]

Suppose $\alpha>1$, in which case $w$ is $S$-shaped (i.e., first convex, and then concave). Then $f''(\lambda^*)<0$; moreover, $f'(\lambda)>0$ for $\lambda<\lambda^*$, and $f'(\lambda)<0$ for $\lambda>\lambda^*$. This implies that $\lambda^*$ is the unique maximizer of \eqref{v reduced}. By the discussion following Corollary 4.6 in \cite{XZ13}, we conclude that
\begin{equation}\label{ttau example}
\hat\tau_x := \inf\{t\ge 0 : X^x_t\ge x/\lambda^*\} = \inf\left\{t\ge 0 : X^x_t\ge e^{\left(\frac{\eta}{\alpha\gamma}\right)^{1/(\alpha-1)}} x\right\}\in\Tc
\end{equation}
is a pre-committed  optimal stopping time of \eqref{v} when the current state is $x$. A {\it moving} stopping threshold causes time-inconsistency: for any $t>0$, $\hat\tau_x\neq t+\hat\tau_{X^x_t}$ on $\{\hat\tau_x>t\}$, unless $X^x_t \equiv x$; thus, \eqref{tidef} is in general violated. More specifically, as $X$ evolves over time, the agent continuously updates the initial value $x$ in \eqref{v} with current state $y:=X_t$, and thus changes the original stopping threshold $x/\lambda^*$ to $y/\lambda^*$ at time $t$. While being an optimal solution to \eqref{v} at state $x$, $\hat\tau_x$ will not be implemented  as future selves will dismiss  it.

A crucial observation here is that time-inconsistency leads the na\"{i}ve agent to postpone stopping indefinitely. As $1/\lambda^*>1$, we have $\hat\tau_x>0$ for all $x>0$, whence the na\"{i}ve agent will never stop at any given  moment and -- as a result -- will {\it never} stop.
\end{Example}

\begin{Remark}
Example \ref{eg:never stop} is reminiscent of \cite{ebert2015until}, which shows -- under the CPT setting -- that a na\"{i}ve agent may defer his stopping decision indefinitely (hence disobey the original pre-committed optimal stopping time) under the so-called ``skewness preference in the small" condition, which hinges on Assumptions 1 and 2 therein.
Indeed, in Example \ref{eg:never stop}, we have
\[
\zeta: =\sup_{x>0}u'_\ell(x)/u'_r(x) = 1\quad \hbox{and}\quad w'(0+)=\infty,
\]
where $u'_\ell$ and $u'_r$ denote the left and right derivatives of $u$, respectively. This shows that Assumptions 1 and 2 in \cite{ebert2015until} are both satisfied: the former requires $\zeta<\infty$, while the latter is fulfilled if $w'(0+)>\zeta$ (see (3) therein). 
\end{Remark}

In the presence of time-inconsistency, it is of interest to study how a sophisticated agent lacking commitment might be doing by applying the general result Theorem \ref{thm:main}.
Specifically, how the fixed-point iteration in Theorem \ref{thm:main} can be applied to find equilibrium stopping laws, under the current context of probability distortion.

Technically, in order to apply Theorem \ref{thm:main} we need to first check the validity of its assumptions.

\begin{Lemma}\label{prop:J satisfies}
Suppose $\beta>0$.
The objective functional \eqref{J} satisfies Assumption~\ref{asm:J}.
\end{Lemma}

The proof of Lemma~\ref{prop:J satisfies} is relegated to Appendix~\ref{appen:J satisfies}.

\begin{Theorem}\label{coro:main}
Suppose $\beta>0$ and the objective functional is given by \eqref{J}.
For any $\tau\in\Tc(\R_+)$, $\tau_*$ defined in \eqref{tau0} belongs to $\Ec(\R_+)$,
and hence \eqref{Ec characterization} holds with $\X=\R_+$.
In particular, $\hat\tau_*$ defined by \eqref{wtau*} belongs to $\Ec(\R_+)$.
\end{Theorem}

\begin{proof}
As $X$ is a geometric Brownian motion with zero drift, the process $Z$ in \eqref{Z} is constantly $1$ and thus a martingale.
Because the objective functional \eqref{J} satisfies Assumption~\ref{asm:J} (by Lemma~\ref{prop:J satisfies}),
the result is now a direct consequence of Theorem~\ref{thm:main}.
\end{proof}

%
%

The fixed-point iteration in Theorem~\ref{coro:main} has intriguing implications. The next result shows that a drastic change in behavior takes place when a na\"{i}ve agent applies the fixed-point iteration {\it just once}: a na\"{i}ve agent who would have never stopped (``until the bitter end" \cite{ebert2015until}) transforms himself into a sophisticated one who stops immediately.

\begin{Proposition}\label{prop:drastic change}
Let $\beta>0$ and assume $u(x)>0$ $\forall x>0$. Then for $\tau\in\Tc(\R_+)$ given by $\tau(x)=1$ $\forall x>0$,
\[
\tau_*(x) = \lim_{n\to\infty} \Theta^n\tau (x) = \Theta\tau (x) =0\quad \forall x>0.
\]
\end{Proposition}
\begin{proof}
First of all $\beta>0$ implies $u$ is nondecreasing and $u(0)=0$.
By the definition of $\tau$, $\Lc^*\tau(x) = \infty$ $\forall x>0$. This, together with \eqref{0 at infinity}, implies $J(x;\Lc^*\tau) = \int_0^\infty w(\P[u(0)>y]) dy=0< u(x)$ $\forall x>0$. It follows that $S_{\tau} = \R_+$, whence $\Theta\tau(x) = 0$ $\forall x>0$.
\end{proof}

\begin{examcont}{eg:never stop}
We have shown that as $\alpha>1$, a na\"{i}ve agent postpones stopping indefinitely, i.e., $\hat\tau(x)=1$ for all $x>0$. By Proposition~\ref{prop:drastic change}, the corresponding equilibrium stopping law $\hat\tau_* = \lim_{n\to\infty}\Theta^n\hat\tau = \Theta\hat\tau$ is the trivial equilibrium described in Remark~\ref{rem:trivial}, namely, to stop immediately.
\end{examcont}

The na\"{i}ve,  
once thinking like a sophisticate, decides to stop immediately. What constitutes the economic reasoning behind the proof of Proposition \ref{prop:drastic change} is the following. At any {\it given} point in time, assume that all the future selves will be carrying out the na\"{i}ve, never-stopping strategy. With this in mind, the question is what to do now. There are only two options. If he is to follow the original strategy, namely, to continue now, then according to \eqref{0 at infinity} the value of the (un-stopped) underlying process will almost surely diminish to zero and hence the payoff is zero. If he is to stop now, then he can  get some positive payoff because of the assumption $u(x)>0$ when $x>0$. This simple comparison will prompt him to stop immediately.

\smallskip

The use of Theorem~\ref{coro:main} will be further demonstrated in the next two subsections. In Subsection \ref{subsec:examples}, we study practical examples motivated by utility maximization and real options valuation. We are particularly interested in finding the equilibrium stopping law  $\hat\tau_*\in\Ec(\R_+)$ generated by the na\"{i}ve stopping law, as defined in \eqref{wtau*}. Through this construction of an equilibrium we will be able to see an intriguing connection between a na\"{i}ve agent and a sophisticated one -- in particular how the former might turn himself into the latter by evoking strategic reasonings. In Subsection~\ref{subsec:demonstration}, we show how the fixed-point iteration in Theorem~\ref{coro:main} can be used to find a large class of equilibrium stopping laws, under a specific choice of the transformed function $u$ and the distortion function $w$. How one can possibly compare these equilibria will also be discussed.

\subsection{Examples: Utility Maximization and Real Options Valuation}\label{subsec:examples}
In this subsection, we will study several examples motivated by utility maximization and real options valuation. In all the examples, we will take the function $U$ to be either $U(x) = x^{\gamma}$ for some $0<\gamma<1$, a standard isoelastic utility function, or $U(x) = (x-K)^+$ for some $K>0$, a common payoff function for real options. Our focus will be to compare the na\"{i}ve stopping law $\hat\tau$ in \eqref{induced policy} with the corresponding equilibrium stopping law $\hat\tau_*$ in \eqref{wtau*}.

We start with assuming that $w$ is inverse $S$-shaped (i.e., concave on $[0, 1-q]$ and convex on $[1-q, 1]$ for some $q\in(0,1)$), which indicates that the agent exaggerates the probabilities of both ``very good'' and ``very bad'' scenarios. This type of $w$ has been widely studied, as it is consistent with empirical data of human decision making\footnote{\cite{HZ16} relate an inverse $S$-shaped $w$ with the emotion of hope (hope for the ``very good'' scenarios) and fear (fear of the ``very bad'' scenarios) in decision making.}. The literature on inverse $S$-shaped $w$, including \cite{TK92}, \cite{LBW92}, \cite{CH94}, \cite{WG96}, \cite{BM96}, and \cite{Prelec98}, among others, suggests three main models of $w$: 1) the one-parameter model
\begin{equation}\label{1-factor}
w(x) = \frac{x^\gamma}{(x^\gamma+(1-x)^{\gamma})^{1/\gamma}}\quad \hbox{with}\ \gamma^*\le\gamma<1,
\end{equation}
where $\gamma^*\approx 0.279$ is the minimal value of $\gamma$ such that $w$ is nondecreasing; 2) the two-parameter model
\begin{equation}\label{2-factor}
w(x) = \frac{\alpha x^\gamma}{\alpha x^\gamma + (1-x)^{\gamma}}\quad \hbox{with}\ \alpha>0,\ 0<\gamma<1;
\end{equation}
and 3) the two-parameter model \eqref{Prelec's w} with $0<\alpha<1$ and $\gamma>0$.

To facilitate discussions in subsequent examples, let us present equivalent expressions for the na\"{i}ve law $\hat\tau$. Suppose optimal stopping times $\{\hat\tau_x\}_{x\in\R_+}$ exist. For each $x>0$, let $F_{\hat\tau_{x}}$ be the cumulative distribution function of $X^x_{\hat\tau_x}$, and $G_{\hat\tau_{x}}:= F_{\hat\tau_{x}}^{-1}$ be the quantile function of $X^x_{\hat\tau_x}$. Then,
the na\"{i}ve stopping law $\hat\tau\in\Tc(\R_+)$ in \eqref{induced policy} can be expressed using $F_{\hat\tau_x}$ or $G_{\hat\tau_x}$:
\begin{equation}\label{induced policy_quantile}
\hat\tau(x):=
\begin{cases}
0,\ \hbox{if}\ \hat\tau_{x}=0,\\
1,\ \hbox{if otherwise}.
\end{cases}
=
\begin{cases}
0,\ \hbox{if}\ F_{\hat\tau_{x}}(\cdot)= 1_{[x,\infty)}(\cdot),\\
1,\ \hbox{if otherwise}.
\end{cases}
=
\begin{cases}
0,\ \hbox{if}\ G_{\hat\tau_{x}}(\cdot) \equiv x,\\
1,\ \hbox{if otherwise}.
\end{cases}
\end{equation}
An interesting consequence is that to identify the na\"{i}ve stopping law $\hat\tau$, it is {\it not} necessary to find out what $\hat\tau_x$ is specifically: in view of \eqref{induced policy_quantile}, $\hat\tau$ is well-defined once the quantile function $G_{\hat\tau_{x}}$ is known. As shown in \cite{XZ13}, $G_{\hat\tau_{x}}$ is the optimal solution to
\begin{equation}\label{supJ(G)}
\sup_{G\in\Qc_x} \int_0^1 u(G(y))w'(1-y) dy,  
\end{equation}
where $\Qc_x := \{G:[0,1)\mapsto\R_+ \mid G\ \hbox{is a quantile function of $X^x_\tau$ for some $\tau\in\Tc$}\}$; moreover, \eqref{supJ(G)} can be solved through certain mathematical programs. This approach, developed in \cite{XZ13} for solving
pre-committed optimal stopping, is called the ``quantile formulation" and useful in the next result.

\begin{Proposition}\label{thm:concave u, R S-shaped w}
Suppose $u$ is concave and $w$ is inverse $S$-shaped with $w'(0+) = \infty$. 
Assume that the optimal quantile $G^*_x$ of \eqref{supJ(G)}, as specified in Theorem 5.2 of \cite{XZ13}, exists for all $x>0$.
\begin{itemize}
\item [(i)] If $\sup_{y>0}u(y)$ is attained, then
\[
\hat\tau(x)= \hat\tau_*(x)=1_{(0,y^*)}(x)\quad \hbox{for all $x>0$},
\]
where $y^* := \inf\{y>0 : u(y)=\sup_{z>0}u(z)\}<\infty$.
\item [(ii)] If $\sup_{y>0}u(y)$ is not attained, then
\begin{equation*}
\begin{split}
\hat\tau(x)=1\quad \hbox{for all $x>0$},\\
\hat\tau_*(x) =0 \quad \hbox{for all $x>0$}.
\end{split}
\end{equation*}
\end{itemize}
\end{Proposition}
The proof of Proposition~\ref{thm:concave u, R S-shaped w} is relegated to Appendix~\ref{sec:proofs for Sec 4.3}.

\begin{Remark}
All the three major forms of $w$ proposed in the literature and supported by empirical evidence, \eqref{1-factor}, \eqref{2-factor}, and \eqref{Prelec's w} with $0<\alpha<1$ and $\gamma>0$, are inverse $S$-shaped with $w'(0+)=\infty$. The result of Proposition~\ref{thm:concave u, R S-shaped w} is therefore of sufficient practical relevance.
\end{Remark}

Proposition~\ref{thm:concave u, R S-shaped w}-(ii) is consistent with the results of \cite{ebert2015until} and \cite{ESnever2017}: the na\"{i}ve agent never stops while the sophisticated one stops immediately. Both \cite{ebert2015until} and \cite{ESnever2017} assume that $u$ is {\it strictly} increasing, which is stronger than the condition that $\sup_{y>0}u(y)$ is not attained. On the other hand, Proposition~\ref{thm:concave u, R S-shaped w}-(i) shows that if there exists a state that maximizes the payoff function $u$ itself, then it makes no sense for even the na\"{i}ve 
to hold the asset forever: he ought to stop whenever such a state is reached. Moreover, because such a threshold type strategy can be upheld by all the selves, it is also a sophisticated strategy. Note that this conclusion also demonstrates that the respective extreme stopping behaviors of the two types of agents reported in \cite{ebert2015until} and \cite{ESnever2017} depend critically on the assumption that the payoff function $u$ is {\it strictly} increasing. When $\sup_{y>0}u(y)$ is  attained which violates the assumption, then both agents will instead adopt the same, threshold-type of strategy.

The above result can be extended to the case where $u$ is $S$-shaped.

\begin{Proposition}\label{thm:concave u, R S-shaped w'}
Suppose $u$ is $S$-shaped (i.e., convex on $[0, \theta]$ and concave on $[\theta, \infty)$ for some $\theta>0$) and $w$ is specified as in Proposition~\ref{thm:concave u, R S-shaped w}.
Assume that the optimal quantile $G^*_x$ of \eqref{supJ(G)}, as specified in Section 6 of \cite{XZ13}, exists for all $x>0$. Then, Proposition~\ref{thm:concave u, R S-shaped w}-(i) and (ii)  still hold under current setting.
\end{Proposition}

The proof of Proposition~\ref{thm:concave u, R S-shaped w'} is relegated to Appendix~\ref{sec:proofs for Sec 4.3}.

Now we present several examples.

\begin{Example}\label{eg:utility max}
Let $U(x) = x^{\gamma}$ for some $0<\gamma<1$, and $w$ be inverse $S$-shaped with $w'(0+)=\infty$. 
In this case, $u(x) = U(x^{1/\beta}) = x^{\gamma/\beta}$.
\begin{itemize}
\item {\bf Case I:} $0<\gamma/\beta \le 1$.\\
As $u$ is concave, we conclude from Proposition~\ref{thm:concave u, R S-shaped w} that $\hat\tau(x) =1$ for all $x>0$, and $\hat\tau_*(x)=0$ for all $x>0$.
\item {\bf Case II:} $\gamma/\beta >1$.\\
As $u$ is convex, as mentioned in Example~\ref{eg:never stop}, the optimal value $\sup_{\tau\in \Tc} J(x;\tau)$ can be computed as \eqref{v reduced} for all $x>0$.
If $w$ is given by \eqref{Prelec's w} with $0<\alpha<1$ and $\gamma>0$. Observe that the optimal value is infinite for all $x>0$:
\begin{equation}\label{optimal=infty}
\sup_{\tau\in\Tc}J(x;\tau)=\lim_{\lambda\downarrow 0} w(\lambda) u\left(\frac{x}{\lambda}\right) = \lim_{\lambda\downarrow 0} \frac{\exp(-\gamma(-\log(\lambda))^\alpha)}{\lambda^\eta}x^\eta = \lim_{y\to\infty} \frac{e^{-\gamma y^\alpha}}{e^{-\eta y}}x^\eta = \infty,
\end{equation}
where the last equality follows from $\alpha<1$. If an optimal stopping time $\hat\tau_x$ exists, we must have $\hat\tau_x>0$ as stopping immediately does not attain the optimal value (as $u(x)<\infty$).
The na\"{i}ve stopping law is therefore to continue perpetually (i.e., $\hat\tau(x)=1$ for all $x>0$), and the corresponding sophisticated behavior is to stop immediately (i.e., $\hat\tau_*(x)=0$ for all $x>0$).

If we instead take $w$ to be \eqref{1-factor} or \eqref{2-factor}, similar calculations again yield $\sup_{\tau\in\Tc}J(x;\tau)=\lim_{\lambda\downarrow 0} w(\lambda) u\left(\frac{x}{\lambda}\right)=\infty$. Thus, we have the same conclusion that $\hat\tau(x)=1$ for all $x>0$ and $\hat\tau_*(x)=0$ for all $x>0$.
\end{itemize}
\end{Example}

\begin{Example}
Let $U(x) = (x-K)^+$ for some $K>0$, and $w$ be inverse $S$-shaped with $w'(0+)=\infty$. 
In this case $u(x) = U(x^{1/\beta}) = (x^{1/\beta}-K)^+$.
\begin{itemize}
\item {\bf Case I:} $0<\beta \le 1$.\\
As $u$ is convex, the optimal value $\sup_{\tau\in\Tc}J(x;\tau)$ can be computed as \eqref{v reduced} for all $x>0$.
If $w$ is given by \eqref{Prelec's w} with $0<\alpha<1$ and $\gamma>0$. Similarly to \eqref{optimal=infty}, the optimal value is infinite for all $x>0$:
\begin{align*}
\sup_{\tau\in\Tc}J(x;\tau)&=\lim_{\lambda\downarrow 0} w(\lambda) u\left(\frac{x}{\lambda}\right) = \lim_{\lambda\downarrow 0} \exp(-\gamma(-\log(\lambda))^\alpha)\left[\left(\frac{x}{\lambda}\right)^{1/\beta}-K\right]\\
& = \lim_{y\to\infty} \left(\frac{e^{-\gamma y^\alpha}}{e^{-(1/\beta) y}}x^{1/\beta}-e^{-\gamma y^\alpha}K\right) = \infty,
\end{align*}
where the last equality follows from $\alpha<1$. Thus, as explained in Case II of Example~\ref{eg:utility max}, the na\"{i}ve stopping law is to continue perpetually (i.e., $\hat\tau(x)=1$ for all $x>0$), and the corresponding sophisticated behavior is to stop immediately (i.e., $\hat\tau_*(x)=0$ for all $x>0$).

If we instead take $w$ to be \eqref{1-factor} or \eqref{2-factor}, similar calculations again yield $\sup_{\tau\in\Tc}J(x;\tau)=\lim_{\lambda\downarrow 0} w(\lambda) u\left(\frac{x}{\lambda}\right)=\infty$. Thus, we have the same conclusion that $\hat\tau(x)=1$ for all $x>0$ and $\hat\tau_*(x)=0$ for all $x>0$.

\item {\bf Case II:} $\beta > 1$.\\
Because $u$ is first constantly zero and then concave, it is $S$-shaped. We then conclude from Proposition~\ref{thm:concave u, R S-shaped w'} that $\hat\tau(x) =1$ for all $x>0$, and $\hat\tau_*(x)=0$ for all $x>0$.
\end{itemize}
\end{Example}

In the above two examples, a na\"{i}ve agent always wants to continue perpetually, which prompts the sophisticated agent to stop immediately (as shown in Proposition~\ref{prop:drastic change}).
It is natural to ask whether a na\"{i}ve agent will indeed choose to stop under a different model parameter
specification. The answer is yes when he has a more pessimistic view prescribed by a convex $w$. A convex $w$ indicates that the agent deflates the probability of the ``very good'' scenarios, and inflates the probability of the ``very bad'' scenarios. This contrasts the optimistic view on the ``very good'' scenarios represented by an inverse $S$-shaped $w$.

With a convex $w$, the following example shows that a na\"{i}ve agent may choose to stop immediately, depending on the intensity  of probability weighting relative to the agent's risk tolerance and the quality of the asset.

\begin{Example}\label{eg:u,w convex power}
Let $U(x) = x^{\gamma}$ for some $0<\gamma<1$, and $w(x)=x^\eta$ for some $\eta>1$. Then, $u(x) = U(x^{1/\beta}) = x^{\gamma/\beta}$.
\begin{itemize}
\item {\bf Case I:} $0<\gamma/\beta \le 1$.\\
As $u$ is concave and $w$ is convex, Corollary 4.3 of \cite{XZ13} shows that the optimal stopping time is $\hat\tau_x=0$ for all $x>0$. The problem is then time-consistent, and all three types of agents (na\"{i}ve, pre-committed, and sophisticated) stop immediately.

Note that $u$ is concave here because $\beta>0$ is not small enough, i.e., the asset is not good enough. With such an asset at hand, the pessimistic view from a convex $w$ leads the agent to liquidate the asset immediately. The moral of this result is that ``never lay your hands on a bad stock if you are pessimistic''.
\item {\bf Case II:} $\gamma/\beta > 1$.\\
As $u$ is convex, we can find optimal stopping times using  \eqref{v reduced}, which now takes the form $\sup_{\lambda\in(0,1]}\lambda^{\eta-\gamma/\beta} x^{\gamma/\beta}$. If $\eta\ge\gamma/\beta$, $\lambda^*=1$ is the maximizer, and thus $\hat\tau_x= \inf\{t\ge 0:X^x_t\ge x\} = 0$ is the optimal stopping time. Hence, there is no time-inconsistency when $\eta\ge\gamma/\beta$, and all three types of agents stop immediately.

The interesting case is $\eta<\gamma/\beta$, in which $\sup_{\lambda\in(0,1]}\lambda^{\eta-\gamma/\beta} x^{\gamma/\beta}=\infty$, not attainable by any $\lambda\in(0,1]$. In this case, an optimal stopping time $\hat\tau_x$ fails to exist for all $x>0$. The na\"{i}ve stopping law is therefore $\hat\tau(x)=1$ for all $x>0$. This leads to the equilibrium $\hat\tau_*(x) = \Theta\hat\tau(x)= 0$ for all $x>0$, thanks to Proposition~\ref{prop:drastic change}.
\end{itemize}
\end{Example}

In the above example, $\eta\equiv w'(1)$ measures the intensity of probability weighting on very unfavorable events,
$1-\gamma$ is the Arrow-Pratt measure of relative risk-aversion of the agent, 
and $\beta$ measures the goodness of the asset (the large $\beta$ is, the worse the asset; recall the explanation below \eqref{beta}). When the agent sufficiently inflates the probabilities of bad scenarios relative to his risk tolerance and the quality of the asset (i.e., $\eta\ge \gamma/\beta$), he liquidates the asset immediately no matter what type of agent he is. Otherwise (i.e., $\eta < \gamma/\beta$), he always intends to continue if he is na\"{i}ve, and stops right away if he is sophisticated.

\smallskip

In fact,  a convex probability distortion function $w$ can generate richer behaviors than the two extremes ``never stop'' and ``stop immediately'' as depicted in most previous examples. The next example shows that both the na\"{i}ve and sophisticated agents may
actually agree on a threshold-type strategy, as in \eqref{bar x threshold} below.


\begin{Example}\label{eg:u,w convex w'(0) finite2}
Let $U(x) = (x-K)^+$ for some $K>0$ and $w(x)= \frac12(x^2+x)$. Then, $u(x)= U(x^{1/\beta})=(x^{1/\beta}-K)^+$. We assume that $\beta\ge 1$, which will be crucial in the calculation below. For any $x>0$, define $\tau_{ab}:=\inf\{t\ge 0: X^x_{t}\notin (a,b)\}$ for $a<b$. First, we claim that for each $x>0$,
\begin{equation}\label{sup=sup}
\sup_{\tau\in\Tc} J(x;\tau)=\sup_{0\le a\le x\le b<\infty} J(x;\tau_{ab}) = \sup_{0\le a\le K^\beta\wedge x,\ K^\beta\vee x\le b<\infty} J(x;\tau_{ab}).
\end{equation}
The first equality was established in Theorem 4.2 of \cite{XZ13}. Now we prove the second equality. For $x\le K^\beta$, observe that if $a\le x\le b\le K^\beta$, then $J(x;\tau_{ab})=0$ as $u(a)=u(b)=0$. This already implies that \eqref{sup=sup} is true. For $x\le K^\beta$, by the convexity of $w$ and the concavity of $u$ on $[K^\beta,\infty)$, the same argument in Corollary 4.3 of \cite{XZ13} shows that
\[
\sup_{K^\beta\le a\le x\le b<\infty} J(x;\tau_{ab})\le u(x).
\]
It follows that \eqref{sup=sup} holds. Now, for any $a\le K^\beta\wedge x$ and $b\ge K^\beta\vee x$,
\begin{equation*}
f(a,b) := J(x;\tau_{ab})=u(a) + w\left(\frac{x-a}{b-a}\right) (u(b)-u(a))= \frac12\left[\left(\frac{x-a}{b-a}\right)^2+\frac{x-a}{b-a}\right](b^{1/\beta}-K),
\end{equation*}
where the second equality follows from Lemma~\ref{lem:tau_ab} and the third equality is due to $u(a)=0$ for $a\le K^\beta$. For any fixed $b\ge K^\beta\vee x$, observe that $f(0,b)\ge f(a,b)$ for all $0<a\le x$. Then a direct calculation yields
\[
\frac{\partial f(0,b)}{\partial b} = \frac{x}{2 b^3} h(b),\quad\hbox{where}\quad h(b):= \left(\frac 1\beta-1\right)b^{\frac 1\beta +1}-\left(\frac 1\beta-2\right)x b^{\frac 1\beta}+Kb+2xK.
\]
We deduce from
\[
h''(b) = -\frac 1\beta \left(1-\frac 1\beta\right) b^{\frac 1\beta-1}\left[\left(\frac 1\beta +1 \right)-\left(2-\frac 1\beta\right) \frac xb\right]
\]
that on $(0,\infty)$, $h$ starts being convex, becomes less convex as $b$ increases, and eventually turns concave. This, together with $h(0)=2xK>0$ and $h(\infty)=-\infty$, shows that there exists a unique $b^*(x)>0$ such that $h(b^*)=0$. As a consequence, $h(x)>0$ if and only if $x<b^*(x)$. Because
\[
h(x) = 
x\left[\left(\frac 2\beta-3\right)x^{\frac 1\beta}+3K\right]
\]
we conclude that $x<b^*(x)$ if and only if $x< \bar x:= \left(\frac{3K}{3-2/\beta}\right)^\beta$. Thus,
\[
\hat\tau_x :=\tau_{0b^*}=\inf\{t\ge 0 : X^x_{t}\ge b^*(x)\}=
\begin{cases}
>0\quad &\hbox{if}\ 0<x< \bar x,\\
=0\quad &\hbox{if}\ x\ge \bar x
\end{cases}
\]
is an optimal stopping time of \eqref{v}, with $J(x;\tau)$ specified in \eqref{J}. Time-inconsistency is present here as the stopping threshold $b^*(x)$ depends on the current state $x$.
 The na\"{i}ve stopping law is
\begin{equation}\label{bar x threshold}
\hat\tau(x) = 1_{\left(0,\bar x\right)}(x),\quad \hbox{for all}\ x>0.
\end{equation}
Remarkably, this is already an equilibrium stopping law. To see this, for $x\ge \bar x$, $\Lc^*\hat\tau(x) = 0$ and thus $J(x;\Lc^*\hat\tau(x))= u(x)$, which yields $[\bar x,\infty)\subseteq I_{\hat\tau}$. For $x<\bar x$, $\Lc^*\hat\tau(x) = \tau_{0\bar{x}}=\inf\{t\ge 0:X^x_t\ge \bar x\}$.
Then
\[
J(x;\Lc^*\hat\tau(x)) = w\left(\frac{x}{\bar x}\right)u(\bar x) =\frac 12\left[\left(\frac{x}{\bar x}\right)^2+\frac{x}{\bar x}\right] (\bar x^{\frac 1\beta}-K)^+,
\]
where the first equality is due to Lemma~\ref{lem:tau_ab}. If $x\le K^\beta$, $J(x;\Lc^*\hat\tau(x))>0 = u(x)$, which implies $(0,K^\beta]\subseteq C_{\hat\tau}$. Now, observe that the curve
\[
y= g(x):= \frac 12\left[\left(\frac{x}{\bar x}\right)^2+\frac{x}{\bar x}\right] (\bar x^{\frac 1\beta}-K)= \frac{u(\bar x)}{2\bar x^2} \left(x+\frac{\bar x}{2}\right)^2-\frac{u(\bar x)}{8}
\]
is a convex quadratic function that intersects the concave function $y= \kappa(x):= x^{1/\beta}-K$ at $x=\bar x$. Moreover, it can be checked that $g'(\bar x) = \kappa'(\bar x)$, which implies that $x=\bar x$ is the only intersection of these curves and $g(x)>\kappa(x)$ on $(0,\bar x)$. If $x\in (K^\beta,\bar x)$, observe that $J(x;\Lc^*\hat\tau(x)) = g(x) > \kappa(x) = u(x)$, which shows that $(K^\beta,\bar x)\subseteq C_{\hat\tau}$. We therefore conclude that $C_{\hat\tau} = (0,\bar x)$ and $I_{\hat\tau}=[\bar x,\infty)$. It follows that $\Theta\hat\tau(x) = 1_{(0,\bar x)}+\hat\tau(x) 1_{[\bar x,\infty)}=\hat\tau(x)$ for all $x>0$. That is, $\hat\tau\in\Ec(\R_+)$, and thus $\hat\tau_* = \hat\tau$.
\end{Example}

\begin{Remark}
The above example is consistent with Proposition~\ref{eg:u,w convex w'(0) finite}.  
Indeed, with $\beta=1$, we have $\bar x = 3K$ in the above example. This coincides with the threshold $\frac{\eta+1}{\eta} K = 3K$ in Proposition~\ref{eg:u,w convex w'(0) finite}, when we take $\eta=1/2$ therein.
\end{Remark}

\subsection{A Case Study: Finding and Comparing Equilibrium Stopping Laws}\label{subsec:demonstration}
In this subsection, we focus on the case where
\begin{equation}\label{u and w}
u(x) = (x-K)^+\quad \hbox{and}\quad w(x)= \eta x^2+(1-\eta)x,
\end{equation}
for some $K>0$ and $\eta\in(0,1)$. Our goal is to demonstrate how fixed-point iteration in Theorem~\ref{coro:main} can be used to find a large class of equilibrium stopping laws. We will also discuss how one can possibly compare these different equilibria.

All the proofs in this subsection are relegated to Appendix~\ref{sec:proofs for Sec 4.2}.

We start with identifying (pre-committed)  optimal stopping times and the corresponding na\"{i}ve stopping law.

\begin{Proposition}\label{eg:u,w convex w'(0) finite}
Assume $u$ and $w$ are given by \eqref{u and w}.
For $x>\big(\frac{1-\eta}{\eta}\wedge 1\big)K$, an optimal stopping time of \eqref{v} is
\[
\hat\tau_x :=
\begin{cases}
\inf\{t\ge 0: X^x_t\ge x\} = 0,\quad &\hbox{if}\ x>\frac{\eta+1}{\eta}K,\\
\inf\{t\ge 0: X^x_t\ge \frac{2\eta Kx}{\eta (x+K)-K}\}>0,\quad &\hbox{if}\ \big(\frac{1-\eta}{\eta}\wedge 1\big)K<x\le \frac{\eta+1}{\eta}K.
\end{cases}
\]
For $x\le \big(\frac{1-\eta}{\eta}\wedge 1\big)K$, there exists no optimal stopping time.
Hence, the na\"{i}ve stopping law is a threshold-type strategy
\begin{equation}\label{naive threshold}
\hat\tau(x) = 1_{\big(0,\frac{\eta+1}{\eta}K\big)}(x),\quad \hbox{for all}\ x>0.
\end{equation}
\end{Proposition}

\begin{Remark}
We observe time-inconsistency in Proposition~\ref{eg:u,w convex w'(0) finite}: the stopping threshold $\frac{2\eta Kx}{\eta (x+K)-K}$ in $\hat\tau_x$ depends on current state $x$, when $\big(\frac{1-\eta}{\eta}\wedge 1\big)K<x\le \frac{\eta+1}{\eta}K$.
\end{Remark}

With \eqref{naive threshold}, it is interesting that even a na\"{i}ve agent 
{\it will} stop, seemingly contradicting \cite{ebert2015until}. However, there is really no contradiction -- the convexity of $w$ violates Assumption 2 of \cite{ebert2015until}. In economic terms, a convex distortion overweights the small probabilities of very bad events and underweights those of very good events; so a small, right-skewed, lottery-like gamble is unattractive to the agent. As a result, he decides that he will stop once the process hits the threshold, $\frac{\eta+1}{\eta}K$. This also suggests that RDU/CPT together with the results derived in this paper can indeed offer realistic predictions, if we allow modeling of different types of preferences and different characteristics of the underlying process.

Given that the na\"{i}ve stopping law is a threshold-type strategy as in \eqref{naive threshold}, the next result looks for equilibrium stopping laws of the same form.

\begin{Proposition}\label{prop:find equilibria}
Assume $u$ and $w$ are given by \eqref{u and w}. For any $b>0$, consider $\tau:= 1_{(0,b)}\in \Tc(\R_+)$. If $b\le \frac{\eta+1}{\eta} K$, then $\tau\in\Ec(\R_+)$. If $b>\frac{\eta+1}{\eta} K$, then $\tau\notin\Ec(\R_+)$ but $\Theta\tau \in \Ec(\R_+)$; specifically, there exists $b'\in(\frac{K}{\eta},\frac{\eta+1}{\eta} K)$ such that
\[
\tau_*(x) = \lim_{n\to\infty}\Theta^n\tau(x) =\Theta\tau(x)=1_{\left(0,b'\right)}(x),\quad  \forall x>0.
\]
Hence, $1_{\left(0,b\right)}\in \Ec(\R_+)$ if and only if $0<b\le\frac{\eta+1}{\eta} K$.
\end{Proposition}

\begin{Remark}
Remarkably, it follows from Propositions \ref{eg:u,w convex w'(0) finite} and \ref{prop:find equilibria} that the na\"{i}ve stopping law $\hat\tau$ in \eqref{naive threshold} is already an equilibrium, i.e., $\hat\tau_* = \hat\tau$.
\end{Remark}

In view of Proposition \ref{prop:find equilibria} and Remark~\ref{rem:trivial}, we have found an uncountable set of equilibrium stopping laws of threshold-type
\[
\Ec'(\R_+) := \left\{1_{(0,b)} :  0\le b\le \left(\frac{\eta+1}{\eta}\right)K\right\},
\]
where $1_{(0,0)}$ denotes the trivial equilibrium specified in Remark~\ref{rem:trivial}. It is natural to ask how one can compare these different equilibria, and whether there exists a {\it best} equilibrium, in any reasonable sense.

One possibility is to investigate the choice from the perspective of the agent at time 0.  Assume that the agent at time 0 intends to choose $\tau\in\Ec'(\R_+)$ to maximize the expected discounted payoff, i.e.,
\begin{equation}\label{Player 0's}
\sup_{\tau\in \Ec'(\R_+)} J(x;\Lc\tau(x)) = \sup_{b\in \big[0,\frac{\eta+1}{\eta} K\big]} J(x;T^x_{[b,\infty)}),\quad \hbox{with $T^x_{[b,\infty)}$ defined in \eqref{T_D}}.
\end{equation}

\begin{Proposition}\label{prop:optimal E}
Assume $u$ and $w$ are given by \eqref{u and w}. For $0<x<\big(\frac{\eta+1}{\eta}\big)K$, $\hat \tau$ in \eqref{naive threshold} is the unique maximizer of \eqref{Player 0's}. For $x\ge \big(\frac{\eta+1}{\eta}\big)K$, $\tau \mapsto J(x;\Lc\tau(x))$ is a constant function on $\Ec'(\R_+)$.
\end{Proposition}

An intriguing implication of Proposition~\ref{prop:optimal E} is that not only the agent at time 0 but also all his future selves would like to use the equilibrium $\hat \tau$ in \eqref{naive threshold}, as it is optimal for \eqref{Player 0's} at {\it every} state $x>0$. This corresponds to the notion of an {\it optimal equilibrium} introduced in \cite{HZ17-discrete}: it is an equilibrium that generates larger value than any other equilibrium does at every state, i.e., a universally dominating equilibrium. For stopping problems under non-exponential discounting, \cite{HZ17-discrete} and \citeyearpar{HZ17-continuous} establish the existence of an optimal equilibrium, under appropriate conditions on the discount function. Under the current setting of probability distortion, it is natural to ask if the result in Proposition~\ref{prop:optimal E} can be extended  to general $u$ and $w$. Finding economically meaningful conditions on these functions under which an optimal equilibrium exists is an interesting direction for future research.

We can further quantify the cost of {\it not} following the optimal equilibrium $\hat\tau$ in \eqref{naive threshold}. Specifically, when the agent is at a given state $x>0$, if he follows an equilibrium $\tau_b := 1_{(0,b)}\in \Ec'(\R_+)$ with $b< \frac{1+\eta}{\eta}K$, the cost $c\ge 0$ is defined as the additional cash needed to raise the value of the inferior equilibrium $\tau_b$ to that of the optimal equilibrium $\hat \tau$. That is, $c\ge 0$ is the solution to
\begin{equation}\label{cost}
\int_0^\infty w\left(\P[u(X^x_{\Lc\tau_b(x)} + c) > y]\right) dy = \int_0^\infty w\left(\P[u(X^x_{\Lc \hat\tau(x)}) > y]\right) dy = J(x,\Lc \hat\tau(x)).
\end{equation}
Note that $c= c(x,b)$ depends on the state $x>0$ and the threshold $b\ge 0$.

\begin{Proposition}\label{prop:cost}
Assume $u$ and $w$ are given by \eqref{u and w}, and let $b^*:=\big(\frac{\eta+1}{\eta}\big)K$. Given $0<x<b^*$, the cost $c=c(x,b)$ admits the formula
\[
c(x,b) = K-(b\vee x) +\frac{K}{\eta}\left(\frac{w\left({x}/{b^*} \right)}{w\left(({x}/{b})\wedge 1\right)}\right)> 0, 
\quad \hbox{for all}\ 0\le b< b^*.
\]
Given $x\ge b^*$, $c(x,b)= 0$ for all $0\le b< b^*$.
\end{Proposition}

Instead of the stringent ``uniform optimality" as above, 
{\it Pareto optimality} has been widely used in the literature to compare different equilibria. It can be formulated in the current context as follows.

\begin{Definition}
For any  $\tau_1, \tau_2 \in\Ec(\R_+)$, we say $\tau_1$ dominates $\tau_2$ if $J(x,\Lc\tau_1(x))\ge J(x,\Lc\tau_2(x))$ for all $x>0$. Furthermore, we say $\tau\in\mathcal A\subseteq\Ec(\R_+)$ is Pareto optimal in $\mathcal A$ if $\tau$ cannot be dominated by any other $\tau'\in\mathcal A$. 
\end{Definition}

\begin{Proposition}\label{prop:Pareto}
Assume $u$ and $w$ are given by \eqref{u and w}. For any $1_{(0,b)}, 1_{(0,b')}\in \Ec'(\R_+)$, $1_{(0,b)}$ dominates $1_{(0,b')}$ if and only if $b\ge b'$. Hence, $\hat \tau$ in \eqref{naive threshold} is the only Pareto optimal equilibrium in $\Ec'(\R_+)$.
\end{Proposition}


\section{Conclusions}
In this paper, we have established a fixed-point characterization of equilibrium stopping laws for time-inconsistent stopping problems. This provides a new, convenient way to search for equilibrium laws, through fixed-point iterations. In addition, it spells out how the na\"{i}ve stopping law $\hat \tau$ (defined in \eqref{wtau*}) can be transformed into an equilibrium law $\hat\tau_*$.

Our framework is general enough to cover many well-known time-inconsistent stopping problems. In particular, we apply our theoretic results to the case where the agent distorts probability and the underlying state process is a geometric Brownian motion. From all the examples we studied, including those in Subsection~\ref{subsec:examples} and the case study in Subsection~\ref{subsec:demonstration}, we observe that under $\beta> 0$ (where time-inconsistency is present), whenever the iteration starts with the na\"{i}ve law $\hat\tau$, we fall in one of the two cases: (a) the na\"{i}ve law $\hat\tau$ is to continue perpetually and the corresponding equilibrium law $\hat\tau_*$ is to stop immediately, and (b) $\hat\tau$ is aligned with $\hat\tau_*$, and they are both threshold-type strategies. Whether this is a general result that can be established theoretically is an interesting topic for future research.


\appendix

\section{Proof of Proposition \ref{prop:ker increases}}\label{appen:T_x=0}
To prove Proposition \ref{prop:ker increases}, we need to first analyze certain crossing probabilities for the one-dimensional diffusion $X$ introduced in Section~\ref{subsec:1-D case}. To this end, we consider the running maximum and minimum processes defined respectively by
\[
\overline{X}^x_t := \max_{s\in[0,t]} X^x_s\quad\hbox{and}\quad \underline{X}^x_t:= \min_{s\in[0,t]} X^x_s,\quad t\ge 0.
\]
Also, we consider the first revisit time to the initial value $x$:
\begin{equation}\label{Tx}
T^x_x := \inf\{t>0 : X^x_t = x\}\in\Tc.
\end{equation}

\begin{Lemma}\label{lem:T_x=0}
Suppose \eqref{integrability} holds and $Z$ defined in \eqref{Z} is a martingale.
Then, for any $x\in\X$, $\P[\overline{X}^x_t>x]=\P[\underline{X}^x_t<x]=1$ for all $t> 0$. Hence, $T^x_x = 0$ a.s.
\end{Lemma}

\begin{proof}
First, recall from Problem 7.18 on p. 94 of \cite{KS-book-91} that if $W$ is a standard Brownian motion defined on a probability space $(\Omega',\Fc',\P')$, then
\begin{equation}\label{step 1}
\P'[\overline{W}_t>0] = \P'[\underline{W}_t<0]=1,\quad \hbox{for all}\ t> 0,
\end{equation}
where $\overline{W}_t := \max_{s\in[0,t]} W_s$ and $\underline{W}_t:= \min_{s\in[0,t]} W_s$.

Fix $T>0$. With $Z$ in \eqref{Z} being a martingale, we can define a probability $\Q\approx\P$ by $\frac{d\Q}{d\P}=Z_T$.
Note that we have $\Q\approx\P$, instead of merely $\Q\ll\P$, because $Z_T>0$ $\P$-a.s. under \eqref{integrability}. Girsanov's theorem then implies that
under $\Q$, $dX_t = a(X_t) d\widetilde{B}_t$ for $t\in[0,T]$, where
$
\widetilde{B}_t := B(t) + \int_0^t\theta(X_s)ds,\ t\in [0,T],
$
is a $\Q$-Brownian motion. As $X$ is a continuous local martingale under $\Q$, it can be expressed as a time-changed Brownian motion, i.e.,
\[
X_t = x+ W_{[X]_t}\quad  t\in[0,T],
\]
for some standard Brownian motion $\{W_t\}_{t\ge 0}$ under $\Q$. As a consequence,
\begin{align*}
\Q[\overline{X}_T\ge x+\eta] &= \Q[X_s\ge x+\eta,\quad \hbox{for some}\ 0\le s\le T]\\
&= \Q[W_{[X]_s}\ge \eta,\quad \hbox{for some}\ 0\le s\le T]\\
&= \Q[W_{s}\ge \eta,\quad \hbox{for some}\ 0\le s\le [X]_T].
\end{align*}
This implies
\begin{equation}\label{lim E[1]}
\Q[\overline{X}^x_T> x] = \lim_{\eta\downarrow 0}\Q[\overline{X}^x_T\ge x+\eta] =  \lim_{\eta\downarrow 0}\E^\Q\left[1_{\{W_{s}\ge \eta,\ \ \text{for some}\ 0\le s\le [X]_T\}}\right].
\end{equation}
In view of \eqref{step 1}, we have $\Q[\overline{W}_t>0]=1$ for all $t> 0$. This implies that we can find some $\Omega^*\in\Fc$ with $\Q(\Omega^*)=1$ such that for each $\omega\in\Omega^*$, there exist a real sequence $\{t_n(\omega)\}$ with $t_n(\omega)\downarrow 0$ and $W_{t_n(\omega)}(\omega)>0$. It follows that for each $\omega\in\Omega^*$,
\[
1_{\{W_{s}\ge \eta,\ \ \text{for some}\ 0\le s\le [X]_T\}}(\omega) = 1,\quad \hbox{as $\eta$ small enough}.
\]
We hence conclude from \eqref{lim E[1]} that $\Q[\overline{X}^x_T> x] =1$. Similarly, we have
\begin{equation}
\Q[\underline{X}^x_T< x] = \lim_{\eta\downarrow 0}\Q[\underline{X}_T\le x-\eta] =  \lim_{\eta\downarrow 0}\E^\Q\left[1_{\{W_{s}\le -\eta,\ \ \text{for some}\ 0\le s\le [X]_T\}}\right] =1,
\end{equation}
where the last equality follows from $\Q[\underline{W}_t<0]=1$ for all $t> 0$, as shown in \eqref{step 1}. With $\Q\approx\P$, we conclude that $\P[\overline{X}^x_T> x] =\Q[\overline{X}^x_T> x] =1$ and $\P[\underline{X}^x_T< x] =\Q[\underline{X}^x_T< x] =1$.

Note that the result ``$\P[\overline{X}^x_T> x] =\P[\underline{X}^x_T< x] =1$ for all $T>0$'' implies that for any $T = 1/n$, $n\in\N$, the process $X^x$ will cross the horizontal line of level $x$ during the time interval $(0,1/n)$ $\P$-a.s. As a result, $T^x_x < 1/n$ $\P$-a.s. for all $n\in\N$, and thus $T^x_x =0$ $\P$-a.s.
\end{proof}

Now, we are ready to prove Proposition~\ref{prop:ker increases}.

\begin{proof}[Proof of Proposition~\ref{prop:ker increases}]
For any $x\in\ker(\tau)$, there are three possible cases:
\begin{itemize}
\item [1.] $x$ is an interior point of $\ker(\tau)$:
Then $\Lc^*\tau(x)=0$ by definition, and thus $J(x;\Lc^*\tau(x))=u(x)$, i.e., $x\in I_\tau$.
\item [2.] $x$ is a boundary point of $\ker(\tau)$:
By Lemma~\ref{lem:T_x=0}, $\P[\overline{X}^x_t> x] =\P[\underline{X}^x_t< x] =1$ for all $t>0$.
This implies $\Lc^*\tau(x)=\inf\{t>0:X^x_t\in\ker(\tau)\}<1/n$ for all $n\in\N$ a.s., and thus $\Lc^*\tau(x)=0$ a.s.
It follows that $J(x;\Lc^*\tau(x))=u(x)$, i.e., $x\in I_\tau$.
\item [3.] $x$ is an isolated point of $\ker(\tau)$, i.e., $\exists\ \eps>0$ such that $(x-\eps,x+\eps)\cap\ker(\tau)=\{x\}$:
As $T^x_x=0$ a.s. (by Lemma~\ref{lem:T_x=0}), $\Lc^*\tau(x)=0$ a.s.
This gives $J(x;\Lc^*\tau(x))=u(x)$, i.e., $x\in I_\tau$.
\end{itemize}
It follows that $\ker(\tau)\subseteq I_\tau$.
This, together with \eqref{Theta}, shows that $\ker(\Theta\tau) = S_\tau\cup(I_\tau\cap\ker(\tau))=S_\tau\cup\ker(\tau)$.
By repeating the same argument above for each $n>1$, we obtain \eqref{ker increases}.

As $\{\ker(\Theta^n\tau)\}_{n\in\N}$ is a nondecreasing sequence of Borel sets, if $x\in\bigcup_{n\in\N} \ker(\Theta^n\tau)$, then there exists $N>0$ such that $\Theta^n\tau(x)=0$ for all $n\ge N$; if $x\notin\bigcup_{n\in\N} \ker(\Theta^n\tau)$, then $\Theta^n\tau(x)=1$ for all $n\in N$.
This already implies that the limit-taking in \eqref{tau0} is well-defined, and $\ker(\tau_*) = \bigcup_{n\in\N} \ker(\Theta^n\tau)$.
\end{proof}

\section{Proof of Lemma \ref{prop:J satisfies}}\label{appen:J satisfies}

To prove Lemma \ref{prop:J satisfies}, we need the following technical result.

\begin{Lemma}\label{lem:tau_ab}
Suppose $\beta >0$. For any $0\le a<x<b$, denote $\tau_{ab}:=\inf\{t\ge 0: X^x_t\notin (a,b)\}\in\Tc$.
\begin{itemize}
\item [(i)] If $a=0$ and $b=\infty$, then $J(x;\tau_{ab}) = 0$;
\item [(ii)] If $a>0$ and $b=\infty$, then $J(x;\tau_{ab})= u(a)$;
\item [(iii)] If $a\ge 0$ and $b<\infty$, then
\begin{equation}\label{J for tau_ab}
J(x;\tau_{ab}) = u(a) + w\left(\frac{x-a}{b-a}\right) (u(b)-u(a)).
\end{equation}
\end{itemize}
\end{Lemma}

\begin{proof}
First note that parts of this lemma were derived in \cite{XZ13}, while the cases where $\P[\tau_{ab}=\infty]>0$ were not dealt with there. This includes ``$a=0$ and $b=\infty$'' and ``$a=0$ and $b<\infty$''. For completeness and reader's convenience, we present the proof for all possible cases  of $0\le a<b\le\infty$.

Recall that $\tau_{ab} = T^x_a\wedge T^x_b$, with $T^x_a$ and $T^x_b$ defined as in \eqref{Tb<infinity}, and $u$ is nondecreasing with $u(0)=0$ when $\beta>0$.

(i) Observe that $\tau_{ab} = \infty$ a.s. By \eqref{0 at infinity}, $J(x;\tau_{ab})=\int_0^\infty w(\P[u(0)>y])dy = 0$.

(ii) Thanks to \eqref{0 at infinity}, $u(X^x_{\tau_{ab}})=u(a)$ a.s. It follows that $J(x;\tau_{ab})=\int_0^\infty w(\P[u(a)>y])dy = u(a)$.

(iii) We first deal with the case ``$a>0$ and $b<\infty$''. The CDF of $X^x_{\tau_{ab}}$ is $F(y) = p^* 1_{[a,b)}(y)+1_{[b,\infty)}(y)$ for $y\in[0,1]$, where $p^*:=\P[X^x_{\tau_{ab}}=a] = \frac{b-x}{b-a}$ by the optional sampling theorem. Note that the use of the  optional sampling theorem to find $p^*$ requires either $\P[T^x_a<\infty]=1$ or $\P[T^x_b<\infty]=1$. The former is true here thanks to $a>0$ and \eqref{0 at infinity}.
By the distribution formulation ((3.1) in \cite{XZ13}),
\begin{align*}
J(x;\tau_{ab}) &= \int_0^a w(1) u'(y) dy +\int_a^b w\left(1-\frac{b-x}{b-a}\right) u'(y) dy + \int_b^\infty w(0) u'(y) dy,
\end{align*}
which yields the desired result as $w(0)=0$ and $w(1)=1$.

For the case ``$a=0$ and $b<\infty$'', by the fact that $X^x$ does not reach $a=0$ a.s. and \eqref{0 at infinity}, $X^x_{\tau_{ab}} = b 1_{\{T^x_b<\infty\}}$ a.s. This, together with $u(0)=0$, gives $u(X^x_{\tau_{ab}}) = u(b)1_{\{T^x_b<\infty\}}$ a.s. It follows that
\begin{align*}
J(x;\tau_{ab}) &= \int_0^\infty w\left(\P[u(b)1_{\{T^x_b<\infty\}}>y]\right) dy = \int_0^{u(b)} w\left(\P[u(b)1_{\{T^x_b<\infty\}}>y]\right) dy\\
&= \int_0^{u(b)} w(\P[T^x_b<\infty]) dy=w(\P[T^x_b<\infty])u(b) = w\left(\frac xb\right)u(b),
\end{align*}
where the last equality follows from \eqref{Tb<infinity}. Thus, the formula \eqref{J for tau_ab} still holds for $a=0$.
\end{proof}


Now, we are ready to prove Lemma \ref{prop:J satisfies}.

\begin{proof}[Proof of Lemma \ref{prop:J satisfies}]
Fix a $D\in\Bc(\R_+)$. For any $x>0$, define
\begin{equation}\label{a(x) b(x)}
a(x):=\sup_{}\{a<x:a\in D\},\qquad b(x):=\inf_{}\{b>x:b\in D\}.
\end{equation}
If $a(x)=x$ or $b(x)=x$, then $T^x_D=0$ a.s. and thus $J(x;T^x_D)= u(x)$. If $a(x)< x<b(x)$, then $T^x_D=\inf\{t\ge 0: X^x_{t}\notin (a(x),b(x))\}$. We then deduce from Lemma~\ref{lem:tau_ab} that
\begin{equation}\label{J Borel}
J(x;T^x_D) =
\begin{cases}
u(x),\quad &\hbox{if}\  a(x)=x\ \text{or}\ b(x)=x,\\
u(a(x)) + w\left(\frac{x-a(x)}{b(x)-a(x)}\right) (u(b(x))-u(a(x))),\quad &\hbox{if}\ a(x)<x<b(x)<\infty,\\
u(a(x)),\quad &\hspace{-0.27in}\hbox{if}\ 0<a(x)<x<b(x)=\infty,\\
0,\quad &\hspace{-0.27in}\hbox{if}\ 0=a(x)<x< b(x)=\infty.
\end{cases}
\end{equation}
Now, note that with $D_\Q:=\{q\in\Q:a\le q\le b\ \hbox{for some}\ a,b\in D\}$,
\[
a(x) = \sup\{q<x:q\in D_\Q\}=\sup_{q\in D_\Q} q1_{(q,\infty)}(x),\quad b(x) = \inf\{q>x:q\in D_\Q\}=\inf_{q\in D_\Q} q 1_{(0,q)}(x).
\]
It follows that $x\mapsto a(x)$ and $b\mapsto b(x)$ are both Borel measurable. We then conclude form \eqref{J Borel} that $x\mapsto J(x;T^x_D)$ is Borel measurable, i.e., Assumption~\ref{asm:J} (i) is satisfied.

Fix a sequence $\{D_n\}_{n\in\N}$ in $\Bc(\R_+)$ such that $D_{n}\subseteq D_{n+1}$ for all $n\in\N$. For any $x>0$, consider $a(x)$ and $b(x)$ as in \eqref{a(x) b(x)} with $D:= \bigcup_{n\in\N} D_n$. Moreover, we define
\begin{align*}
a_n(x) & :=\sup_{}\{a<x:a\in D_n\},\quad b_n(x):=\inf_{}\{b>x:b\in D_n\},\qquad \hbox{for all $n\in\N$}.
\end{align*}
As $D_{n}\subseteq D_{n+1}$ for all $n\in\N$, we have $a_n(x)\uparrow a(x)$ and $b_n(x)\downarrow b(x)$ as $n\to\infty$. For each $n\in\N$, by the same argument above \eqref{J Borel},
\begin{equation*}
J(x;T^x_{D_n}) =
\begin{cases}
u(x),\quad &\hbox{if}\  a_n(x)=x\ \text{or}\ b_n(x)=x,\\
u(a_n(x)) + w\left(\frac{x-a_n(x)}{b_n(x)-a_n(x)}\right) (u(b_n(x))-u(a_n(x))),\quad &\hbox{if}\ a_n(x)<x<b_n(x)<\infty,\\
u(a_n(x)),\quad &\hspace{-0.27in}\hbox{if}\ 0<a_n(x)<x<b_n(x)=\infty,\\
0,\quad &\hspace{-0.27in}\hbox{if}\  0=a_n(x)<x< b_n(x)=\infty.
\end{cases}
\end{equation*}
From the above formula and \eqref{J Borel}, one may deduce from the continuity of $u$ and $w$ and the convergence $a_n(x)\uparrow a(x)$ and $b_n(x)\downarrow b(x)$ that
\begin{equation}\label{Dn to D}
\lim_{n\to\infty}J(x;T^x_{D_n}) = J(x;T^x_{D}).
\end{equation}
Indeed, the only nontrivial case is ``$a(x) = b(x)=x$ while $a_n(x)< x < b_n(x)$ for all $n\in\N$''. In this case,
\begin{align*}
\lim_{n\to\infty}J(x;T^x_{D_n}) &= \lim_{n\to\infty}\left[u(a_n(x)) + w\left(\frac{x-a_n(x)}{b_n(x)-a_n(x)}\right) (u(b_n(x))-u(a_n(x)))\right]\\
&= u(x) + \lim_{n\to\infty} \left[w\left(\frac{x-a_n(x)}{b_n(x)-a_n(x)}\right) (u(b_n(x))-u(a_n(x)))\right]= u(x)=J(x;T^x_{D}),
\end{align*}
where the third equality follows from $a(x) = b(x)=x$ and $w$ being a bounded function. Then \eqref{Dn to D} in particular implies that Assumption~\ref{asm:J} (ii) is satisfied.
\end{proof}


\section{Proofs for Subsection~\ref{subsec:examples}}\label{sec:proofs for Sec 4.3}
\begin{proof}[Proof of Proposition~\ref{thm:concave u, R S-shaped w}]
Thanks to \eqref{induced policy_quantile}, for each $x>0$, $\hat\tau(x)=0$ if and only if $G^*_{x}(\cdot)\equiv x$.

(i) With $y^* <\infty$, we observe from $w'(0+) = \infty$ that for any $\lambda\ge 0$, the term $(u')^{-1}_\ell\left(\frac{\lambda}{w'(1-y)}\right)$ in (5.6) of \cite{XZ13} increases to $y^*$ as $y\uparrow 1$.  Then, for any $x>0$, (5.7) of \cite{XZ13} implies that $G^*_{x}(\cdot)\equiv x$ holds only when $a^*=x\ge y^*$. This already shows that $\hat\tau(x)=1$ for all $x<y^*$. Now, for $x\ge y^*$, observe that with $a=x$, the constraint in (5.6) of \cite{XZ13} becomes
\[
x q +\int_q^1 x\vee(u')^{-1}_\ell\left(\frac{\lambda}{w'(1-y)}\right)dy = x q +\int_q^1 xdy= x.
\]
That is, the constraint is satisfied for any $\lambda\ge 0$. Similarly, with $a=x$ and any $\lambda\ge 0$, the objective function in (5.6) of \cite{XZ13} has the value
\[
(1-w(1-q))u(x) + \int_q^1 u\left(x \right) w'(1-y)dy = u(x)=\sup_{y>0} u(y).
\]
This already attains the maximum of (5.6) in \cite{XZ13}. Indeed, the maximization therein can be simplified as
\begin{align*}
&\max_{a,\lambda\ge 0} \left\{(1-w(1-q))u(a) + \int_q^1 u\left(a\vee (u')^{-1}_\ell\left(\frac{\lambda}{w'(1-y)}\right) \right) w'(1-y)dy\right\}\\
=&\max_{a\ge y^*} \left\{(1-w(1-q))u(a) + \int_q^1 u\left(a\right) w'(1-y)dy\right\} = \max_{a\ge y^*} u(a) =\sup_{y>0} u(y).
\end{align*}
We therefore conclude that $a^*=x$ and any $\lambda\ge 0$ form a solution to (5.6) of \cite{XZ13}, which yields $G^*_x(\cdot)\equiv x$. Thus, $\hat\tau(x) =0$ for all $x\ge y^*$. We therefore obtain $\hat\tau(x) = 1_{(0,y^*)}(x)$ for all $x> 0$.
Finally, for any $x\ge y^*$, $\Lc^*\hat\tau(x)=0$ and thus $J(x;\Lc^*\hat\tau(x))=u(x)$, which implies $\Theta\hat\tau(x)=\hat\tau(x)=0$. For any $x< y^*$, $\Lc^*\hat\tau(x)$ is the first hitting time of $X^x$ to the value $y^*$. Thus, $J(x;\Lc^*\hat\tau(x))=\int_0^\infty w(\P[u(y^*)>y])dy = u(y^*) \ge u(x)$. It follows that $\Theta\hat\tau(x)=1=\hat\tau(x)$. As a result, $\hat\tau_*(x) = \Theta\hat\tau(x) = \hat\tau(x)$ for all $x>0$.

(ii) 
If $\lambda^*=0$, then $(u')^{-1}_\ell\left(\frac{\lambda^*}{w'(1-y)}\right) =(u')^{-1}_\ell(0)=\infty$ in (5.6) of \cite{XZ13} for all $y\in(q,1)$, as $\sup_{y>0}u(y)$ is not attained. This implies that the constraint in (5.6) of \cite{XZ13} cannot be satisfied, a contradiction.
Assume $\lambda^*>0$. Because $w'(0+)=\infty$, $\frac{\lambda^*}{w'(1-y)}\to 0$ as $y\uparrow 1$. The concavity of $u$ and the fact that $\sup_{y>0}u(y)$ is not attained then imply that $(u')^{-1}_\ell\left(\frac{\lambda^*}{w'(1-y)}\right)\to \infty$ as $y\uparrow 1$, which shows $G^*_x(\cdot)\not\equiv x$. Thus, $\hat\tau(x)=1$ for all $x>0$. By Proposition~\ref{prop:drastic change}, $\hat\tau_*(x) = 0$ for all $x>0$.
\end{proof}

\begin{proof}[Proof of Proposition~\ref{thm:concave u, R S-shaped w'}]
In \cite{XZ13}, while the mathematical program for $G^*_x$ in Section 6 (for $S$-shaped $u$) is more complicated than that in Theorem 5.2 (for concave $u$), we note that under the condition $G^*_x(\cdot)\equiv x$ they actually coincide. Thus, the same argument in the proof of Proposition~\ref{thm:concave u, R S-shaped w} can be used to establish the desired results.
\end{proof}


\section{Proofs for Subsection~\ref{subsec:demonstration}}\label{sec:proofs for Sec 4.2}

\begin{proof}[Proof of Proposition~\ref{eg:u,w convex w'(0) finite}]
As $u$ is convex, \eqref{v} is equivalent to \eqref{v reduced}, given by $\sup_{\lambda\in(0,1]} f(\lambda)$ with $f(\lambda):=  (\eta\lambda^2+(1-\eta)\lambda)\left(\frac x\lambda-K\right)^+$. For $x> K$, $f'(\lambda) = \eta x-(1-\eta) K-2\eta K\lambda$. Define
\[
\bar\lambda:=\frac{\eta(x+K)-K}{2\eta K}.
\]
If $x> \frac{\eta+1}{\eta}K$, $f'(\lambda)> (\eta+1)K-(1-\eta)K-2\eta K =0$ for all $\lambda\in(0,1]$, which implies that $\lambda^*=1$ is the maximizer. If $K< x\le \frac{\eta+1}{\eta}K$, then $\lambda^*=\bar\lambda\le 1$ is the maximizer.
For $x\le K$, observe that
\begin{equation}\label{f(lambda)}
f(\lambda)=
\begin{cases}
-\eta K\left[\lambda-\bar\lambda\right]^2 + (1-\eta)x+ \eta K \bar\lambda^2, \quad &\hbox{if}\ 0<\lambda\le \frac xK,\\
0,\quad &\hbox{if}\  \frac xK< \lambda\le 1.
\end{cases}
\end{equation}
If $\eta\le 1/2$, then $\bar\lambda \le \frac{x-K}{4\eta K}\le 0$. This, together with \eqref{f(lambda)}, implies that $\sup_{\lambda\in(0,1]}f(\lambda)$ admits no maximizer. If $\eta>1/2$, then for $x>\frac{1-\eta}{\eta}K$, it can be checked that $\bar\lambda\in (0,\frac xK)$, which implies that $\lambda^*=\bar\lambda$ is the maximizer (thanks again to \eqref{f(lambda)}); for $x\le \frac{1-\eta}{\eta}K$, we have $\bar\lambda\le 0$ and thus $\sup_{\lambda\in(0,1]}f(\lambda)$ admits no maximizer.
Summarizing these cases leads to the formula of $\hat\tau_x$. The na\"{i}ve stopping law follows directly from the formula of $\hat\tau_x$.
\end{proof}

\begin{proof}[Proof of Proposition~\ref{prop:find equilibria}]
For $b\le K$, $J(x;\Lc^*\tau(x)) =u(x)$ for all $x>0$, and thus $I_\tau = \R_+$. This implies $\Theta\tau(x) = 1_{(0,b)}(x) = \tau(x)$ for all $x>0$, i.e., $\tau\in\Ec(\R_+)$.
In the rest of the proof, we assume $b> K$. For $x\ge b$, $\Lc^*\tau(x)=0$ and thus $J(x;\Lc^*\tau(x))=u(x)$, which implies $[b,\infty)\subseteq I_\tau$. For $0<x\le K$, $J(x;\Lc^*\tau(x))>0=u(x)$ and thus $(0,K]\subseteq C_\tau$.
For $K<x< b$, Lemma~\ref{lem:tau_ab}-(iii) gives
\begin{align}
J(x;\Lc^*\tau(x)) -u(x) &= \left(\eta\frac{x^2}{b^2}+(1-\eta) \frac{x}{b}\right) (b-K) - (x-K)\notag\\
& = \frac{\eta(b-K)}{b^2} \left[x^2 - \left(b+\frac{Kb}{\eta(b-K)}\right) x+\frac{Kb^2}{\eta(b-K)} \right]\notag\\
&= \frac{\eta(b-K)}{b^2} (x-b) \left(x- b' \right),\quad  \hbox{where}\ b':= \frac{Kb}{\eta(b-K)}>\frac{K}{\eta}. \label{J-u}
\end{align}

For $b>\frac{\eta+1}{\eta}K$, observe that
\begin{equation}\label{b'<}
b'< \frac{\eta+1}{\eta} K.
\end{equation}
Indeed, $b'/(\frac{\eta+1}{\eta}K)=\frac{b}{(\eta+1)(b-k)}<1$ if and only if $0<\eta(b-K)-K$, which is equivalent to $b>\frac{\eta+1}{\eta}K$. 
Because \eqref{b'<} in particular implies $b' < b$, we deduce from \eqref{J-u} that $J(x;\Lc^*\tau(x)) -u(x)>0$ for $K<x<b'$, $J(x;\Lc^*\tau(x)) -u(x)<0$ for $b'<x<b$, and $J(x;\Lc^*\tau(b')) -u(b')=0$. We therefore conclude that $S_\tau = (b',b)$, $C_\tau = (0,b')$, and $I_\tau = \{b'\}\cup[b,\infty)$, which implies
\[
\tau_1:= \Theta\tau(x) = 1_{(0,b']}(x)\quad \hbox{for all}\ x>0.
\]
By the same argument above \eqref{J-u}, we obtain $[b',\infty)\subseteq I_{\tau_1}$, $(0,K]\subseteq C_{\tau_1}$, and when $K<x< b'$,
\begin{align*}
J(x;\Lc^*\tau_1(x)) -u(x)
&= \frac{\eta(b'-K)}{(b')^2} (x-b') \left(x- b'' \right),\quad \hbox{with}\quad b'':= \frac{Kb'}{\eta(b'-K)}.
\end{align*}
As \eqref{b'<} yields $\frac{K}{\eta(b'-K)}>1$, we have $b''>b'$. It follows that $J(x;\Lc^*\tau_1(x)) -u(x)>0$ for all $K<x<b'$. We thus conclude that $C_{\tau_1} = (0,b')$ and $I_{\tau_1} = [b'',\infty)$, which shows that
\[
\Theta\tau_1(x) = 1_{(0,b']}(x)=\tau_1(x)\quad \hbox{for all}\ x>0.
\]
That is, $\Theta\tau=\tau_1\in\Ec(\R_+)$ and $\tau_*(x) = \lim_{n\to\infty}\Theta^n\tau(x) = \Theta\tau(x) = 1_{(0,b']}(x)$ for all $x>0$. 

Finally, for $b\in (K,\frac{\eta+1}{\eta} K]$, we have $b'\ge \frac{\eta+1}{\eta} K$, by the same argument below \eqref{b'<}. With $b'\ge b$, we deduce from \eqref{J-u} that $C_{\tau} = (0,b)$ and $I_{\tau} = [b,\infty)$. Thus, $\Theta\tau(x) = 1_{(0,b)}(x) = \tau(x)$ for all $x>0$, i.e., $\tau\in\Ec(\R_+)$.
\end{proof}

\begin{proof}[Proof of Proposition~\ref{prop:optimal E}]
For any $x< \big(\frac{\eta+1}{\eta}\big)K$, we deduce from Lemma~\ref{lem:tau_ab} (iii) that
\begin{align*}
\sup_{b\in [0,\frac{\eta+1}{\eta} K]}  J(x;T^x_{[b,\infty)}) &= \sup_{b\in [x,\frac{\eta+1}{\eta} K]} \left(\eta\frac{x^2}{b^2}+(1-\eta) \frac{x}{b}\right) (b-K)= \sup_{b\in [x,\frac{\eta+1}{\eta} K]} x g(1/b) + (1-\eta) x, 
\end{align*}
where $g(y) := -\eta K x y^2+[\eta(x+K)-K] y$. The quadratic function $g(y)$ attains a unique global maximum on $\R$ at $\hat y = \frac{\eta(x+K)-K}{2\eta K x}$. If $x\le \big(\frac{1-\eta}{\eta}\big)K$, then $\hat y\le 0$. This implies $b\mapsto g(1/b)$ is strictly increasing on $[x,\frac{\eta+1}{\eta} K]$, and the maximum is attained at $b^*= \frac{\eta+1}{\eta} K$. If $x> \big(\frac{1-\eta}{\eta}\big)K$, then $\hat y >0$. It can be checked that $x< \big(\frac{\eta+1}{\eta}\big)K$ if and only if $1/\hat y > \big(\frac{\eta+1}{\eta}\big)K$. This again implies that $b\mapsto g(1/b)$ is strictly increasing on $[x,\frac{\eta+1}{\eta} K]$, and the maximum is attained at $b^*= \frac{\eta+1}{\eta} K$.

For any $x\ge \big(\frac{\eta+1}{\eta}\big)K$, we simply have $J(x;T^x_{[b,\infty)})=u(x)$ for all $b\in \big[0,\frac{\eta+1}{\eta} K\big]$.
\end{proof}

\begin{proof} [Proof of Proposition~\ref{prop:cost}]
Fix $0<x<b^*$. By the proof of Proposition~\ref{prop:optimal E} (see above), we have $J(x,\Lc \hat\tau(x)) = w\left(\frac{x}{b^*}\right) \frac{K}{\eta}$.
Given $0\le b< b^*$, let $\tau_b := 1_{(0,b)}\in \Ec'(\R_+)$. If $b\le x$, then $\Lc\tau_b(x) = 0$. Thus, the left hand side of \eqref{cost} becomes $\int_0^\infty w\left(\P[(x+c-K)^+ > y]\right) dy = (x+c-K)^+$. Solving \eqref{cost} then yields
\begin{equation}\label{c1}
c = K-x + w\left(\frac{x}{b^*}\right) \frac{K}{\eta}.
\end{equation}
Note that $c>0$. 
If $x\le K$, $c> 0$ by definition. If $x>K$, we deduce again from $\Lc\tau_b(x) = 0$ that
\[
x-K = \int_0^\infty w\left(\P[u(X^x_{\Lc\tau_b(x)}) > y]\right) dy = J(x, \Lc\tau_b(x)) < J(x,\Lc \hat\tau(x)) = w\left(\frac{x}{b^*}\right) \frac{K}{\eta},
\]
where the inequality follows from Proposition~\ref{prop:optimal E}. This shows that $c>0$.

It remains to deal with the case $b> x$. First, we observe that
\begin{align}
w\left(\frac{x}{b^*}\right) \frac{K}{\eta}-w\left(\frac{x}{b}\right) b &= \left[\eta \left(\frac{x}{b^*}\right)^2+(1-\eta) \frac{x}{b^*}\right]\frac{K}{\eta} - \left[\eta \left(\frac{x}{b}\right)^2+(1-\eta) \frac{x}{b}\right] b\notag\\
&=\eta \left[ \left(\frac{x}{b^*}\right)^2 \frac{K}{\eta} - \left(\frac{x}{b}\right)^2 b \right] + (1-\eta) \left[ \frac{x}{b^*} \frac{K}{\eta}+ x\right]\notag\\
&= \eta x^2 \left[ \frac{1}{(1+\eta)b^*}-\frac{1}{b}\right] - \frac{\eta (1-\eta) x}{1+\eta}<0, \label{w<w}
\end{align}
where the third equality follows from $b^* = \frac{\eta+1}{\eta} K$, and the inequality is due to $x>0$, $\eta\in (0,1)$, and $0<b<b^*$. Now, because $b> x$ implies $\Lc\tau_b(x) = T^x_b$ and  $u(X^x_{T^x_b} + c ) = u(c) 1_{\{T^x_b =\infty\}} + u(b+c) 1_{\{T^x_b < \infty\}}$, the argument in the last equation of the proof of Lemma~\ref{lem:tau_ab} (see Appendix~\ref{appen:J satisfies}) gives
\begin{align}
\int_0^\infty &w\left(\P[u(X^x_{\Lc\tau_b(x)} + c) > y]\right) dy = \int_0^{u(c)} w(1) dy+ \int_{u(c)}^{u(b+c)} w\left(\P[T^x_b <\infty]\right) dy\notag\\
 &= u(c) + w\left(\frac{x}{b}\right) (u(b+c)-u(b)) = (c-K)^+ + w\left(\frac{x}{b}\right) ((b+c-K)^+-(c-K)^+). \label{cost left}
\end{align}
We claim that $c<K$. Indeed, if $c\ge K$, solving \eqref{cost}, with the aid of \eqref{cost left}, yields $c = K+w\left(\frac{x}{b^*}\right) \frac{K}{\eta}-w\left(\frac{x}{b}\right) b$. By \eqref{w<w}, this implies $c<K$, a contradiction. Knowing that $c<K$, \eqref{cost left} yields $\int_0^\infty w\big(\P[u(X^x_{\Lc\tau_b(x)} + c) > y]\big) dy =w\left(\frac{x}{b}\right) (b+c-K)^+$. Solving \eqref{cost} thus gives
\begin{equation}\label{c2}
c = K-b + \frac{K}{\eta} \frac{w\left({x}/{b^*} \right)}{w\left({x}/{b}\right)}.
\end{equation}
Note that $c<K$ can be easily verified: by \eqref{w<w}, $[w\left({x}/{b^*} \right){K}/{\eta}]/  [ w\left({x}/{b}\right) b] < 1$, which implies $c = K - b\big(1 - [w\left({x}/{b^*} \right){K}/{\eta}]/  [ w\left({x}/{b}\right) b]\big) < K$. Also, $c$ is positive: if $b\le K$, then $c> 0$ by definition; if $b>K$, then
$
w\left({x}/{b}\right)(b-K)  = J(x, \Lc\tau_b(x)) < J(x,\Lc \hat\tau(x)) = w\left({x}/{b^*}\right) {K}/{\eta},
$
where the first equality follows from Lemma~\ref{lem:tau_ab} (iii) and the inequality is due to Proposition~\ref{prop:optimal E}. This shows that $c>0$. Combining \eqref{c1} and \eqref{c2} then gives the desired result.

Given $x\ge b^*$, because $\tau \mapsto J(x;\Lc\tau(x))$ is a constant function on $\Ec'(\R_+)$ (Proposition~\ref{prop:optimal E}), $c(x,b) = 0$ for all $0<b<b^*$.
\end{proof}

\begin{proof}[Proof of Proposition~\ref{prop:Pareto}]
The proof of Proposition~\ref{prop:optimal E} shows that $b\mapsto J(x;T^x_{[b,\infty)})$ in nondecreasing on $[0,\frac{\eta+1}{\eta}K]$, for any $x>0$. This directly implies the first assertion. It follows that $1_{(0,b)}$ can be Pareto optimal in $\Ec'(\R_+)$ only when $b$ takes the largest possible value $\frac{\eta+1}{\eta}K$.
\end{proof}

\bibliographystyle{agsm}
\bibliography{refs}

\end{document}